\documentclass[10pt,journal,twocolumn]{IEEEtran}

\IEEEoverridecommandlockouts
\usepackage{cite}
\usepackage{amsmath,amssymb,amsfonts,booktabs}
\usepackage{graphicx}
\usepackage{textcomp}
\usepackage{xcolor}
\usepackage{enumerate}
\usepackage{epstopdf}
\usepackage{subcaption}
\usepackage{amsthm}

\newtheorem{lemma}{Lemma}
\newtheorem{remark}{Remark}
\newtheorem{proposition}{Proposition}

\usepackage{multirow} 

\definecolor{mygray}{RGB}{240,240,240}
\usepackage{colortbl}

\usepackage{algorithm}
\usepackage{algorithmic}

\allowdisplaybreaks[4]

\begin{document}

\title{Movable-Antenna Enabled Robust Vehicular Consumer Networks Under Imperfect CSI}

\author{
        Xuhui Zhang,~\IEEEmembership{Member,~IEEE,}
        Chunjie Wang,~\IEEEmembership{Graduate Student Member,~IEEE,}\\
        Wenchao Liu,~\IEEEmembership{Graduate Student Member,~IEEE,}
        Huijun Xing,~\IEEEmembership{Graduate Student Member,~IEEE,}\\
        Jinke Ren,~\IEEEmembership{Member,~IEEE,}
        Zheng Xing,~\IEEEmembership{Member,~IEEE,}
        and
        Yanyan Shen,~\IEEEmembership{Member,~IEEE}


\thanks{
Xuhui Zhang is with Shenzhen Institutes of Advanced Technology, Chinese Academy of Sciences, Guangdong 518055, China, and also with the Shenzhen Future Network of Intelligence Institute, the School of Science and Engineering, and the Guangdong Provincial Key Laboratory of Future Networks of Intelligence, The Chinese University of Hong Kong (Shenzhen), Guangdong 518172, China (e-mail: xu.hui.zhang@foxmail.com).
}

\thanks{
Chunjie Wang is with Shenzhen Institutes of Advanced Technology, Chinese Academy of Sciences, Guangdong 518055, China, and also with the University of Chinese Academy of Sciences, Beijing 100049, China (e-mail: cj.wang@siat.ac.cn).
}

\thanks{
Wenchao Liu is with the School of Automation and Intelligent Manufacturing, Southern University of Science and Technology, Guangdong 518055, China. (e-mail: wc.liu@foxmail.com).
}

\thanks{
Huijun Xing is with the Department of Electrical and Electronic Engineering, Imperial College London, London SW7 2AZ, The United Kingdom (e-mail: huijunxing@link.cuhk.edu.cn).
}

\thanks{
Jinke Ren is with the School of Science and Engineering, the Shenzhen Future Network of Intelligence Institute, and the Guangdong Provincial Key Laboratory of Future Networks of Intelligence, The Chinese University of Hong Kong (Shenzhen), Guangdong 518172, China (e-mail: jinkeren@cuhk.edu.cn).
}

\thanks{
	Zheng Xing is with the College of Computer Science and Software Engineering, Shenzhen University, Guangdong 518060, China. (e-mail: zhengx@szu.edu.cn)
}

\thanks{
Yanyan Shen is with Shenzhen Institutes of Advanced Technology, Chinese Academy of Sciences, Guangdong 518055, China (e-mail: yy.shen@siat.ac.cn).
}

}

\maketitle

\begin{abstract}
The accelerating advancement of intelligent transportation systems has established consumer-oriented vehicular networks (CVNs) as a critical infrastructure for next-generation connected mobility. 
However, the high mobility of vehicular users (VUs) introduces significant channel state information (CSI) uncertainty, which severely undermines the performance of conventional fixed-position antenna systems. 
To address this, this paper explores the deployment of movable-antennas (MAs) to enhance communication robustness in CVNs under imperfect CSI conditions. 
We develop a joint optimization framework that dynamically coordinates the spatial positioning of MAs and transmit beamforming at the base station, with the objective of maximizing the worst-case sum rate across all VUs. 
The problem is formulated as a non-convex max–min optimization problem, subject to bounded CSI estimation errors, transmit power limits, and physical constraints on antenna displacement. 
By adopting an alternating optimization strategy, the original problem is decomposed into tractable subproblems, solved via techniques including the S-Procedure, Schur complement, and successive convex approximation. 
Numerical evaluations confirm that the proposed approach achieves substantial gains over existing benchmarks in terms of worst-case throughput.
\end{abstract}


\section{Introduction}

\IEEEPARstart{T}{he} accelerating advancement of intelligent transportation systems has elevated vehicular networks to a pivotal role in next-generation mobility, redefining vehicle-to-vehicle and vehicle-to-infrastructure interactions \cite{9509294, 9779322}. 
Among these, consumer-oriented vehicular networks (CVNs) have witnessed rapid adoption, driven by automakers and technology firms that embed wireless connectivity into mainstream automotive platforms \cite{11359731}. 
Such networks empower diverse use cases, including dynamic route guidance, real-time traffic alerts, in-cabin multimedia services, remote vehicle diagnostics, and over-the-air firmware upgrades \cite{10133894}, thereby enriching both driver usability and passenger engagement \cite{9146378}. 
The emergence of sixth-generation (6G) communication standards further intensifies the requirement for ultra-reliable, low-latency wireless links capable of sustaining high-data-rate services during vehicular motion \cite{10012421, 11374087}. 
Consequently, the automotive sector is shifting from standalone onboard systems toward cloud-integrated architectures, demanding wireless infrastructures that are not only scalable and secure but also adaptive to dynamic mobility patterns \cite{10287319}. 
Multi-antenna technologies, leveraging spatial multiplexing and directional beamforming, offer a pathway to concurrently serve multiple mobile users \cite{wang2025joint}. 
Their integration into CVNs represents a foundational enabler for the broader vision of intelligent and autonomous mobility ecosystems, balancing user experience with technological robustness \cite{10874187, 11396947}.

In CVNs, the inherent mobility of vehicles poses a critical limitation to conventional fixed-position antenna (FPA) architectures \cite{liu2025uav}. 
FPAs are typically engineered for quasi-static propagation conditions and often struggle to sustain reliable connectivity under high-velocity motion or in complex urban environments characterized by deep fading and multipath shadowing \cite{11391523, 11328802}. 
The resulting rapid fluctuations in channel state information (CSI), coupled with frequent inter-base station (BS) handovers, lead to significant signal-to-interference-plus-noise (SINR) degradation and elevated packet loss, particularly detrimental to latency-sensitive, high-throughput downlink services such as real-time video streaming or autonomous driving assistance \cite{liu2025movable}. 
To mitigate these challenges, recent efforts have explored movable-antenna (MA) systems, wherein antenna elements are dynamically repositioned, mechanically or electronically, to align with the direction of the target receiver \cite{10416363, 10654366, 11374005}. 
Yet, the performance of MA-enabled systems heavily relies on precise knowledge of vehicle location and orientation, a requirement that becomes increasingly fragile under aggressive acceleration, lane changes, or in global positioning system-challenged zones such as tunnels or dense urban canyons. 
Therefore, robust optimization frameworks must be integrated into CVN designs to explicitly account for spatial uncertainty and imperfect CSI estimation \cite{10663924}.

{
In practical vehicular communication scenarios, particularly within dense urban environments, the accurate and real-time acquisition of CSI remains a fundamental challenge \cite{9355403}. The difficulty stems from the confluence of high-speed vehicle mobility, rapid time-varying multi-path propagation due to dynamic scattering and Doppler-induced channel fluctuations, frequent non-line-of-sight conditions due to buildings and infrastructure, and dynamic shadowing caused by moving obstacles (e.g., other vehicles or pedestrians).
Conventional beamforming and resource allocation schemes, which assume perfect or static CSI, become highly vulnerable under such conditions, leading to significant throughput loss, increased outage probability, and degraded quality of service for safety-critical applications \cite{9110587}.
To mitigate these risks, a robust optimization framework is essential, one that explicitly models CSI uncertainty, (e.g., via bounded error sets), and designs advanced transmission strategies, (e.g., beamforming and antenna position design), to guarantee minimum performance even in the worst-case realization of the uncertain channel \cite{9293148}.
Such an approach not only enhances reliability but also aligns with the stringent latency and dependability requirements of the next-generation CVNs.
}

Motivated by the imperative to enhance the data transmission reliability in highly dynamic vehicular environments, we propose a novel CVN architecture in which a BS equipped with MAs serves multiple single-antenna vehicular users (VUs). 
In contrast to the conventional systems assuming perfect CSI of VUs, our framework explicitly incorporates channel uncertainty stemming from high-velocity motion of VUs or sensor measurement errors.
{
The main contribution of this paper can be summarized as follows:
\begin{itemize}
    \item We consider a novel CVN, where an MA-enabled BS is deployed to support data services for VUs.
    To ensure robust communications under the CSI uncertainty of VUs, we formulate a worst-case VU data rate maximization problem.
    \item To solve this problem, we first convert the objective function with CSI uncertainty into an equivalent deterministic objective function with deterministic constraints, thereby transforming the stochastic optimization problem into a deterministic counterpart.
    \item Subsequently, we develop an efficient alternating optimization (AO) algorithm to solve the problem by iteratively optimizing the transmit beamforming of the BS and the positions of the MAs. Then, we analyze the convergence behavior and computational complexity of the proposed algorithm.
    \item Numerical results confirm that the proposed scheme substantially outperforms the conventional FPA scheme in worst-case throughput, offering a theoretical framework for deploying the MA-enabled systems in the next-generation CVNs.
\end{itemize}}

{
\textit{Organizations:}
The rest of this paper is structured as follows.
First, Section II provides a review of existing literature relevant to CVNs and MAs.
Then, Section III introduces the MA-enabled CVN under uncertain channel conditions and formulates the worst-case data rate maximization problem.
In Section IV, we analyze the problem’s characteristics and develop an efficient AO algorithm to solve it.
Section V presents numerical results to evaluate the performance of the proposed scheme.
Finally, Section VI summarizes the key findings and concludes the paper.
}

{
\textit{Notations:}
In this paper, we adopt the following standard mathematical notations.
The imaginary unit is denoted by $ \mathrm{j} $, satisfying $ {\mathrm{j}^2} = -1 $.
For any complex scalar $z$, $ \Re \{ z \} $ denotes its real part.
Given a matrix $ {\bf{G}} $, its conjugate transpose and ordinary transpose are represented by $ {{\bf{G}}^{\mathsf{H}}} $ and $ {{\bf{G}}^{\mathsf{T}}} $, respectively.
For a vector $ {\bf{w}} $, $ ||{\bf{w}}|| $ stands for its Euclidean.
The notation $ {\cal C}{\cal N}(\mu ,{\sigma ^2}) $ refers to a circularly symmetric complex Gaussian random variable with mean $ \mu $ and variance $ {\sigma ^2} $.
}

\section{Related Works}

{
In this section, we systematically review the state-of-the-art literature through three interconnected research directions.
We first survey recent progress in CVNs, then examine the MA-enabled wireless networks, and finally analyze the robust optimization frameworks tailored for wireless networks, with each critical to addressing the challenges of the next-generation CVNs.
}

\subsection{Recent Advances in CVNs}
{
In recent years, the rapid advancement of 6G wireless networks, mobile computing architectures, and artificial intelligence (AI) technologies has significantly empowered the transformation and upgrading of CVNs by enabling intelligent, ultra-reliable, and low-latency connectivity for the next-generation connected and autonomous vehicles.
}

{
For instance, a vehicular edge computing (VEC) system was studied \cite{9583590}, where the computation efficiency of VUs was maximized under the constraints of the computational latency and energy consumption.
In \cite{10012694}, a digital-twin-assisted VEC system was investigated, and the utility of the VUs was optimized by a deep reinforcement learning (DRL) approach.
With the integration of roadside units, the heterogeneous resource allocation of a VEC system was studied \cite{10354525}, where a long-term task computation problem was tackled by a Lyapunov-driven DRL algorithm.
With connected aerial vehicles (AVs), an integrated space-aerial-ground network was investigated \cite{10946517}, where the utility of data processing efficiency was optimized through the coordination of multi-layer resources.
To facilitate federated learning (FL) over wireless networks, an AV was dispatched to collect local models from ground devices \cite{10972043}, and the latency of each FL round was minimized for efficient model training.
A similar long-term task computation system was studied \cite{10736570}, where the authors proposed a diffusion-enhanced DRL algorithm for resource allocation.
In \cite{10980172}, a multi-AV-enabled data collection system was investigated, and the freshness of the collected data from ground devices was optimized by a novel federated DRL algorithm.
Moreover, an AV was dispatched as a low-altitude platform (LAP) \cite{11298206}, providing information transmission, target sensing, and task computation services for an intelligent CVN.
}

{
Despite recent progress, the majority of existing works on CVNs focus on single-antenna FPA systems or static multi-antenna FPA systems, which inherently limit spatial degrees of freedom (DoFs) and fail to meet the exponentially growing data rate demands of the next-generation vehicular applications, such as autonomous driving, vehicular computing, and real-time road-condition analysis.
Hence, there is an urgent need to embrace reconfigurable multi-antenna architectures, capable of dynamically adapting to high-mobility urban environments, to jointly support ultra-reliable communication, data computing offloading, and integrated sensing and communication functionalities in future CVNs.
}

\subsection{MA-Enabled Wireless Networks}
{
To overcome the limitations of the FPA architectures, recent research has turned to the MA-enabled systems to dynamically align beam patterns and enable novel DoFs for joint optimization of communication, antenna position, and resource allocation with fast-moving vehicles and mitigate channel variability in complex urban environments.
}

{
Firstly, the MA-enabled wireless networks were introduced in \cite{9264694, 9650760}, and the rate capacity was proved by the DoFs from the flexible antenna position optimization.
In addition, the full beampattern gain can be realized over the desired direction by the joint optimization of antenna weights and positions \cite{10278220}.
By combining MAs with an AV-enabled LAP, \cite{10693833}, the information transmission and target sensing were enhanced to enable efficient applications for data processing and intrusion detection.
A secure MA-enabled information and power transfer system was studied \cite{11108293}, where the energy harvesting was enhanced by the joint optimization of the transmit beamforming, antenna position, and artificial noise for ensuring physical layer security (PLS).
To further enhance PLS with uncertainty of eavesdropper, the secure rate of an MA-enabled BS was optimized \cite{11156108}.
In \cite{11240557}, the MA-mounted AV-enabled data collection system was investigated, and the user uplink rates were optimized to facilitate the data computing tasks.
Moreover, a general energy efficiency maximization problem was studied in an MA-enabled multi-user multiple-input multiple-output system, and several novel designs with low complexities were proposed \cite{11261377}.
The MA-enhanced cooperative non-orthogonal multiple access (NOMA) system was investigated \cite{11389911}, and the outage probability of end users was optimized.
}

{
While MAs have demonstrated significant potential to enhance spectral efficiency and spatial agility in future wireless networks, especially the next-generation CVNs, the majority of existing studies assume perfect CSI, which is fundamentally at odds with the high mobility, rapid channel fluctuations, and frequent blockages inherent in real-world urban vehicular scenarios.
Consequently, there is an urgent requirement to investigate the impact of CSI uncertainty on MA-enabled CVNs, and to develop robust transmission strategies that ensure reliable performance even when precise CSI is unavailable.
}

\subsection{Robust Optimization for Wireless Networks}

{
As urban CVNs evolve toward ultra-reliable, low-latency, and AI-driven services, encompassing autonomous driving, cooperative perception, and real-time edge computing, robust optimization becomes indispensable to mitigate the performance degradation caused by imperfect or outdated CSI, particularly in high-mobility, interference-rich urban environments.
}

{
Recently, robust optimization was introduced with several emerging systems.
For instance, the robust beamforming was optimized in a NOMA donwlink CVN \cite{9151971} to enhance the outage performance.
A reconfigurable intelligent surface (RIS)-assisted system was investigated \cite{9184012} with user location uncertainty, and the transmit power was minimized to ensure the worst-case transmission quality.
The secure communication system was studied \cite{9845399} with uncertain CSI and hardware impairments (HIs), where the robust resource allocation was designed to improve worst-case secure rate.
In \cite{10027173}, a dual-function radar-communication (RC) system was investigated, where the radar mismatch and CSI error were jointly addressed through robust optimization.
Moreover, the issue of uncertain CSI and HIs in a RC system was further addressed \cite{10513353}, where the sensing performance was optimized to ensure energy and data transmission requirements.
Meanwhile, some recent works have begun to address channel uncertainty in MA-enabled systems.
In \cite{11192669}, the angular uncertainty was considered in an MA-enabled downlink system, and the total energy efficiency was optimized.
Combining RIS and MAs, a novel symbiotic radio system was investigated \cite{11107317}, where the robust transmission was designed for maximizing the primary rate and ensuring the secondary rate quality.
}

{
Alghough robust optimization has proven effective in addressing the critical challenge of imperfect CSI in urban wireless systems, particularly for highly mobile CVNs, existing works that integrate MAs often assume static or quasi-static user positions, thereby neglecting the dynamic nature of vehicle mobility.
The limitation renders them ill-suited for next-generation CVNs that demand ultra-reliable, high-throughput communication to support computation-intensive applications such as autonomous driving assistance and real-time cooperative perception, motivating the urgent need for a unified robust MA-enabled framework that jointly optimizes beamforming and antenna positioning in CVNs.
}

\section{System Model and Problem Formulation}
{
In this section, we first introduce the MA-enabled CVN system under imperfect CSI of VUs. We then formulate the worst-case rate maximization problem for all VUs and analyze its characteristics.
}

\begin{figure}[t]
\centering
\fbox{\includegraphics[width=0.75\linewidth]{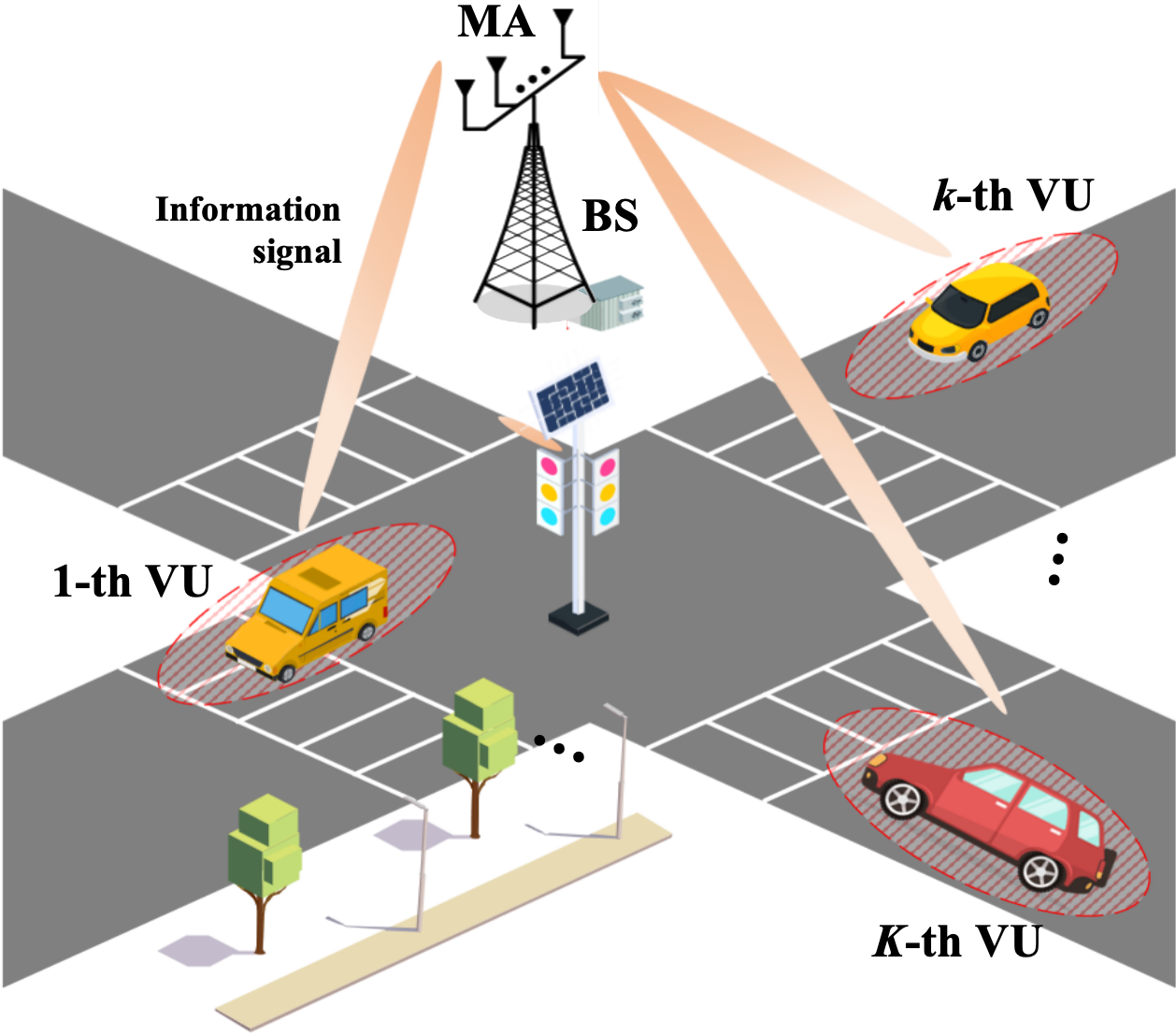}}
\caption{The MA-enabled CVN system.}
\label{sm}
\end{figure}

\subsection{System Model}

We consider a CVN system in which a BS serves $K$ VUs whose positions are uncertain over time,
{as illustrated in Fig. \ref{sm}.}
Let $\mathcal{K} = \{1,2,\ldots,K\}$ denote the index set of VUs.
The BS is equipped with a linear array comprising $M$ transmit MAs, indexed by the set $\mathcal{M} = \{1,2,\ldots,M\}$.
The total service period spans $T$ seconds, partitioned into $N$ discrete time slots, each of length $\tau = \frac{T}{N}$ seconds. The set of all time slots is represented as $\mathcal{N} = \{1,2,\ldots,N\}$.

We adopt a three-dimensional coordinate system to model the spatial positions of the BS and VUs. The BS is fixed at location $\boldsymbol{B} = \{ \boldsymbol{b}, H\}$, where $\boldsymbol{b} = \{ b_x, b_y\}$ represents its horizontal coordinates, and $H$ denotes its deployment altitude.
Throughout the service interval, VUs traverse the roadway. The position of the $k$-th VU at time slot $n$ is expressed as $\boldsymbol{Q}_k[n] = \{ \boldsymbol{q}_{k}[n], 0\}$, with $\boldsymbol{q}_{k}[n] = \{ q_k^x [n], q_k^y [n] \}$ indicating its horizontal coordinates at that instant.

Let $\boldsymbol{x}[n] = \{ x_1[n], \ldots, x_M[n] \}$ represent the position of the MA array at time slot $n$, with $x_m[n]$ indicating the position of the $m$-th MA. We impose the constraint $0 \leq x_1[n] \leq \cdots \leq x_M[n] \leq L$, where $L$ denotes the maximum allowable displacement range for the MAs.
The true steering vector from the BS to the $k$-th VU is then expressed as
\begin{equation} 
    {
    \boldsymbol{a}_k(\boldsymbol{x}[n], \theta_k[n]) = \left[
        \mathrm{e}^{\mathrm{j}\frac{2\pi}{\lambda}x_1[n]\cos\theta_k[n]},
        \ldots,
        \mathrm{e}^{\mathrm{j}\frac{2\pi}{\lambda}x_M[n]\cos\theta_k[n]}
    \right]^\mathsf{H},}
\end{equation}
where $\lambda$ is the wavelength, and $\theta_k[n]$ denotes the elevation component of the angle of arrival (AoA) towards the $k$-th VU, which is given by
\begin{equation} {
    \theta_k[n] = \arccos \frac{H}
    {\sqrt{\Vert \boldsymbol{b}-\boldsymbol{q}_k[n]\Vert^2+H^2}}.}
\end{equation}
Accordingly, the true channel vector $\boldsymbol{h}_k[n]$ from the BS to the $k$-th VU is given by
\begin{equation} {
    \boldsymbol{h}_k(\boldsymbol{x}[n], \theta_k[n]) = \frac{\sqrt{g_0}}{\sqrt{\Vert \boldsymbol{b}-\boldsymbol{q}_k[n]\Vert^2+H^2}} \boldsymbol{a}_k(\boldsymbol{x}[n], \theta_k[n]),}
\end{equation}
where $g_0$ denotes the reference path-loss. Let $s_k[n]$ and $\boldsymbol{w}_k[n]$ denote the transmit signal and the transmit beamforming from the BS to the $k$-th VU at time slot $n$, respectively. Then, the received signal at the $k$-th VU can be expressed as
\begin{equation} {
\begin{split}
    y_k(\boldsymbol{x}[n],& \theta_k[n], \boldsymbol{w}_k[n]) = \boldsymbol{h}_k^{\mathsf{H}} (\boldsymbol{x}[n], \theta_k[n]) \boldsymbol{w}_k[n] s_k[n] \\&+
    \sum_{j=1,j\neq k}^K \boldsymbol{h}_k^{\mathsf{H}} (\boldsymbol{x}[n], \theta_k[n]) \boldsymbol{w}_j[n] s_j[n] + n_k[n],
\end{split}}
\end{equation}
where $n_k \sim \mathcal{CN}(0, \sigma_k^2)$ is the additive white Gaussian noise at the $k$-th VU. As a result, the received SINR of the $k$-th VU at time slot $n$ can be written as
\begin{equation} {
    \gamma_k(\boldsymbol{x}[n], \theta_k[n], \boldsymbol{w}_k[n]) = \frac{\vert \boldsymbol{h}_k^{\mathsf{H}} (\boldsymbol{x}[n], \theta_k[n]) \boldsymbol{w}_k[n] \vert^2}{I_k(\boldsymbol{x}[n], \theta_k[n]) + \sigma_k^2},}
\end{equation}
where $I_k[n]$ denotes the interference of the signal toward other VUs transmitted by the BS, and is given by
\begin{equation} {
    I_k(\boldsymbol{x}[n], \theta_k[n]) = \sum_{j=1,j\neq k}^K \vert \boldsymbol{h}_k^{\mathsf{H}}(\boldsymbol{x}[n], \theta_k[n]) \boldsymbol{w}_j[n] \vert^2.}
\end{equation}
Then, the received data rate at the $k$-th VU at time slot $n$ is
\begin{equation} {
    r_k(\boldsymbol{x}[n], \theta_k[n], \boldsymbol{w}_k[n]) = \log_2(1+\gamma_k(\boldsymbol{x}[n], \theta_k[n], \boldsymbol{w}_k[n])).}
\end{equation}

Due to VUs moving on the road, the BS does not have perfect CSI of VUs. To model this uncertainty, we approximate the true channel as
\begin{equation} {
    \boldsymbol{h}_k(\boldsymbol{x}[n], \theta_k[n]) = \hat{\boldsymbol{h}}_k[n] + \boldsymbol{e}_k[n],
    \label{channel_model}}
\end{equation}
where $\hat{\boldsymbol{h}}_k[n]$ denotes the CSI estimated at the BS that is perfectly available, while $\boldsymbol{e}_k[n]$ captures the associated estimation error. This error is modeled under a bounded uncertainty assumption as
\begin{equation} {
    \Vert \boldsymbol{e}_k[n] \Vert \leq \xi_k, \quad \forall k \in \mathcal{K},\ \forall n \in \mathcal{N},
    \label{error_bound}}
\end{equation}
with $\xi_k$ representing the bound of the uncertainty ball for the $k$-th VU \cite{9676676}.

\subsection{Problem Formulation}
To mitigate performance degradation caused by imperfect CSI, we aim to maximize the time-averaged worst-case data rate across all VUs over the service period, by jointly optimizing the transmit beamforming $\boldsymbol{w}[n] = \{ \boldsymbol{w}_1[n], \ldots, \boldsymbol{w}_K [n]\}, \forall n$, and the position of MAs $\boldsymbol{x}[n], \forall n$. 
The corresponding optimization problem is formulated as
\begin{subequations} { 
\begin{flalign}
 (\textbf{P1}):\ \max_{\boldsymbol{w},\boldsymbol{x}} \quad & \frac{1}{N}\sum_{n= 1}^{N} \left( \min_{{\|\boldsymbol{e}_k[n]\| \leq \xi_k}} \sum_{k= 1}^{K}  r_k(\boldsymbol{x}[n], \theta_k[n], \boldsymbol{w}_k[n])  \right) \nonumber\\
 {\rm{s.t.}}  \quad & \sum_{k= 1}^{K} \Vert \boldsymbol{w}_k[n] \Vert^2 \le P_{\max},\ \forall n\in \mathcal{N},   \label{p1a}\\
 & \{x_m[n]\}_{m=1}^M \in [0,L],\ \forall n\in \mathcal{N},\label{p1b}\\
 &{ x_{p} [n] - x_{q} [n]  \geq d_{\min}, \ \text{where }\{p, q\} \in \mathcal{M}} , \nonumber\\
 &\quad\quad\quad\quad\quad\quad\quad\quad\quad\quad\quad
 { p > q,\ \forall n\in \mathcal{N} },\label{p1c}
\end{flalign}}\end{subequations}where $d_{\min}$ represents the minimum distance between two adjacent MAs.
Constraint \eqref{p1a} ensures that the transmit power of the BS cannot exceed the power upper bound $P_{\max}$.
Constraint \eqref{p1b} denotes the movable region of all MAs.
Constraint \eqref{p1c} imposes that the distance between any two adjacent MAs cannot be lower than $d_{\min}$ for avoiding the coupling effect.
\begin{remark}
    Problem \textbf{P1} is non-convex and cannot be solved by conventional methods.
    Specifically, the achievable rate is a logarithm of a ratio of quadratic forms with respect to the transmit beamforming $\boldsymbol{w}$ and the antenna position $\boldsymbol{x}$,  rendering the objective function non-convex.
    Then, the spherical uncertainty set induces a max-min structure in the objective function, which is non-convex.
    Consequently, the problem \textbf{P1} is an NP-hard non-convex optimization problem, for which no polynomial-time global solver exists.
\end{remark}

\section{Proposed Solution for Robust Optimization}
In this section, we first identify that the CSI uncertainty primarily arises from the VU position estimation errors, which are exacerbated by the high mobility of VUs. 
To tackle this challenge in CVNs, we develop a robust AO framework to jointly design the transmit beamforming and the antenna positions. 
Specifically, for given MA positions, we derive the beamforming under channel uncertainty constraint. 
Subsequently, with beamforming fixed, we optimize the positions of MAs to strengthen system resilience against positional fluctuations.
{
Finally, we review the overall algorithm, and analyze the convergence behavior and computational complexity of the proposed algorithm to show the deployment feasibility.
}

\subsection{Problem Analysis}

The BS employs localization techniques, such as the extended Kalman filter, to estimate the positions of VUs. 
Let $\hat{\boldsymbol{q}}_k[n]$ denote the estimated location of the $k$-th VU at time slot $n$, subject to a bounded positioning error $\Vert \boldsymbol{q}_k[n] - \hat{\boldsymbol{q}}_k[n] \Vert \leq r_k$, where $r_k$ represents the bound of position error.
Accordingly, the corresponding estimated channel is expressed as
\begin{equation} {
    \hat{\boldsymbol{h}}_k[n]=\frac{\sqrt{g_0}}{\sqrt{
        \Vert \boldsymbol{b}-\hat{\boldsymbol{q}}_k[n]\Vert^2+H^2
    }} \hat{\boldsymbol{a}}_k(\boldsymbol{x}[n], \hat{\theta}_k[n]),
    \label{eq:estimated_channel}}
\end{equation}
where $\hat{\theta}_k[n] = \arccos \frac{H}
    {\sqrt{\Vert \boldsymbol{b}-\hat{\boldsymbol{q}}_k[n]\Vert^2+H^2}}$ denotes the estimated AoA.
We next introduce the following proposition to characterize the resulting channel estimation error bound.
\begin{proposition}
    The estimation error $\boldsymbol{e}_k[n] $ towards the $k$-th VU at time slot $n$ is bounded by
    \begin{equation} {
    \begin{split}
        \Vert \boldsymbol{e}_k[n] \Vert &\leq \frac{\sqrt{Mg_0}r_k}{d_k[n]\hat{d}_k[n]}
        \left(
        1+\frac{2\pi L H}{\lambda \hat{d}_k[n]}
        \right)  = \xi_k.
    \label{propek}
    \end{split}}
    \end{equation}
    \label{propno1}
\end{proposition}
\begin{proof}
{
Please refer to Appendix \ref{apppro1}.
}
\end{proof}

\subsection{Problem Transformation}
To transform the \textsc{max-min} form of the objective function of \textbf{P1}, we introduce the auxiliary variables $\tilde{\gamma}_k[n], \forall k, n$. Then, problem \textbf{P2} can be reformulated as
\begin{subequations} { 
\begin{flalign}
 (\textbf{P2}):\ \max_{\boldsymbol{w},\boldsymbol{x}} \quad & \frac{1}{N} \sum_{n= 1}^{N}\sum_{k= 1}^{K}  \tilde{r}_k[n]  \nonumber\\
 {\rm{s.t.}}  \quad & \min_{\boldsymbol{e}_k[n]} \gamma_k(\boldsymbol{x}[n], \theta_k[n], \boldsymbol{w}_k[n]) \geq \tilde{\gamma}_k[n],\nonumber\\
 &\quad \quad \quad \quad \quad \quad \quad \forall k \in \mathcal{K}, \ 
 \forall n \in \mathcal{N},\label{p21a}\\
 &\eqref{p1a}\text{-}\eqref{p1c},\nonumber
\end{flalign}}\end{subequations}where $\tilde{r}_k[n] = \log_2(1+\tilde{\gamma}_k[n])$, and constraint \eqref{p21a} ensures the reliable communication under the worst-case channel estimation error, which can be expanded as
\begin{equation} {
\begin{split}
    \min_{\|\boldsymbol{e}_k[n]\| \leq \xi_k} \frac{\vert \boldsymbol{h}_k^{\mathsf{H}} (\boldsymbol{x}[n], \theta_k[n]) \boldsymbol{w}_k[n] \vert^2}{\sum_{j\neq k}^K \vert \boldsymbol{h}_k^{\mathsf{H}}(\boldsymbol{x}[n], \theta_k[n]) \boldsymbol{w}_j[n] \vert^2 + \sigma_k^2} \geq \tilde{\gamma}_k[n], \\
    \forall k \in \mathcal{K}, \ 
 \forall n \in \mathcal{N}.
\end{split}
\label{channelexpand}}
\end{equation}
Substituting \eqref{channel_model} into \eqref{channelexpand}, we can obtain
\begin{equation} {
\begin{split}
    \min_{\|\boldsymbol{e}_k[n]\| \leq \xi_k} \frac{|(\hat{\boldsymbol{h}}_k [n] + \boldsymbol{e}_k[n])^\mathsf{H} \boldsymbol{w}_k[n]|^2}{\sum_{j \neq k}^K |(\hat{\boldsymbol{h}}_k[n] + \boldsymbol{e}_k[n])^\mathsf{H} \boldsymbol{w}_j[n]|^2 + \sigma_k^2} \geq \tilde{\gamma}_k[n],\\
    \forall k \in \mathcal{K}, \ 
 \forall n \in \mathcal{N}.
\end{split}
\label{channelexpand2}}
\end{equation}
To transform \eqref{channelexpand2}, we introduce the following slack variables $\boldsymbol{\alpha} [n] = \{ \alpha_1[n], \ldots,\alpha_K[n]\}, \forall n$ and $\boldsymbol{\beta} [n] = \{ \beta_1[n], \ldots,\beta_K[n]\}, \forall n$. Then, we can obtain
\begin{equation} {
\begin{split}
    |(\hat{\boldsymbol{h}}_k [n] + \boldsymbol{e}_k[n])^\mathsf{H} \boldsymbol{w}_k[n]|^2 \geq \alpha_k[n],\\
     \|\boldsymbol{e}_k[n]\| \leq \xi_k,\ \forall k \in \mathcal{K}, \ 
 \forall n \in \mathcal{N},
\end{split}
\label{uncertainties1}}
\end{equation}
\begin{equation} {
\begin{split}
    \sum_{j \neq k}^K |(\hat{\boldsymbol{h}}_k[n] + \boldsymbol{e}_k[n])^\mathsf{H} \boldsymbol{w}_j[n]|^2 + \sigma_k^2 \leq
    \beta_k [n],    \\
     \|\boldsymbol{e}_k[n]\| \leq \xi_k,\ \forall k \in \mathcal{K}, \ 
 \forall n \in \mathcal{N},
\end{split}
\label{uncertainties2}}
\end{equation}
\begin{equation} {
    \frac{\alpha_k[n]}{\beta_k[n]} \geq \tilde{\gamma}_k[n],\ \forall k \in \mathcal{K}, \ 
 \forall n \in \mathcal{N}.
 \label{ineq1}}
\end{equation}
To transform the uncertainties in \eqref{uncertainties1} and \eqref{uncertainties2}, we introduce the following Lemma \cite{9676676}.
\begin{lemma}
    (S-Procedure): Let $\boldsymbol{e} \in \mathbb{C}^N $ be a complex vector, and consider the quadratic function
    \begin{equation} {
        f(\boldsymbol{e}) = \boldsymbol{e}^\mathsf{H} \boldsymbol{A} \boldsymbol{e} + 2\Re\left\{\boldsymbol{b}^\mathsf{H} \boldsymbol{e}\right\} + c,}
    \end{equation}
    where \( \boldsymbol{A} \in \mathbb{H}^{N \times N} \), \( \boldsymbol{b} \in \mathbb{C}^N \), and \( c \in \mathbb{R} \). If there exists a scalar \( \lambda \geq 0 \) such that the following block matrix is positive semi-definite, as
    \begin{equation} {
        \begin{bmatrix} \boldsymbol{A} + \lambda \boldsymbol{I} & \boldsymbol{b} \\ \boldsymbol{b}^\mathsf{H} & c - \lambda \xi^2 \end{bmatrix} \succeq \boldsymbol{0},}
    \end{equation}
then \( f(\boldsymbol{e}) \geq 0 \) holds for all \( \boldsymbol{e} \) satisfying \( \|\boldsymbol{e}\| \leq \xi \), where \( \xi > 0 \) is a given constant.
\label{lemmas}
\end{lemma}
Then, we can rewrite \eqref{uncertainties1} as
\begin{equation} {
\begin{split}
    &\boldsymbol{e}_k^\mathsf{H}[n] \boldsymbol{w}_k^\mathsf{}[n]
    \boldsymbol{w}_k^\mathsf{H}[n]
    \boldsymbol{e}_k^\mathsf{}[n] +
    2\Re \{ \hat{\boldsymbol{h}}_k^\mathsf{H}[n]
    \boldsymbol{w}_k^\mathsf{}[n]
    \boldsymbol{w}_k^\mathsf{H}[n]
    \boldsymbol{e}_k^\mathsf{}[n] \} \\ &+
    \hat{\boldsymbol{h}}_k^\mathsf{H}[n]
    \boldsymbol{w}_k^\mathsf{}[n]
    \boldsymbol{w}_k^\mathsf{H}[n]\hat{\boldsymbol{h}}_k^\mathsf{}[n]
    -\alpha_k[n] \geq 0.
\end{split}
\label{uncertainties3}}
\end{equation}
According to Lemma \ref{lemmas}, the inequality \eqref{uncertainties3} can be written as
\begin{equation} {
    \begin{bmatrix} \boldsymbol{w}_k^\mathsf{}[n]
    \boldsymbol{w}_k^\mathsf{H}[n] + \lambda_k[n]\boldsymbol{I} &  \boldsymbol{w}_k^\mathsf{}[n]
    \boldsymbol{w}_k^\mathsf{H}[n] \hat{\boldsymbol{h}}_k^\mathsf{}[n] \\ \hat{\boldsymbol{h}}_k^\mathsf{H}[n]\boldsymbol{w}_k^\mathsf{}[n]
    \boldsymbol{w}_k^\mathsf{H}[n] & \kappa_k[n]  \end{bmatrix}
        \succeq \boldsymbol{0},
        \label{LMI1}}
\end{equation}
where $\boldsymbol{\lambda}[n] = \{ \lambda_1[n],\ldots,\lambda_K [n]\}, \forall n$ are the introduced slack variables, and
\begin{equation} {
    \kappa_k[n] = \hat{\boldsymbol{h}}_k^\mathsf{H}[n]
    \boldsymbol{w}_k^\mathsf{}[n]
    \boldsymbol{w}_k^\mathsf{H}[n]\hat{\boldsymbol{h}}_k^\mathsf{}[n] - \alpha_k[n] - \lambda_k[n]\xi_k^2.
    \label{kappa}}
\end{equation}
To handle \eqref{uncertainties2}, we introduce the following Lemma \cite{10666854}.
\begin{lemma}
        (Schur Complement):
        For a block matrix \( \boldsymbol{M} \in \mathbb{C}^{(N+M) \times (N+M)} \) defined as
        \begin{equation} {
            \boldsymbol{M} = \begin{bmatrix} \boldsymbol{A} & \boldsymbol{B} \\ \boldsymbol{B}^\mathsf{H} & \boldsymbol{C} \end{bmatrix},}
        \end{equation}
where
\( \boldsymbol{A} \in \mathbb{C}^{N \times N} \) is Hermitian and invertible,
\( \boldsymbol{C} \in \mathbb{C}^{M \times M} \) is Hermitian,
and \( \boldsymbol{B} \in \mathbb{C}^{N \times M} \), $\boldsymbol{M} \succeq \boldsymbol{0}$ if and only if $\boldsymbol{A} \succeq \boldsymbol{0}$ and $\boldsymbol{C} - \boldsymbol{B}^\mathsf{H} \boldsymbol{A}^{-1} \boldsymbol{B} \succeq \boldsymbol{0}$.
\label{lemma2}
    \end{lemma}
Then, according to Lemma \ref{lemma2}, we can rewrite \eqref{uncertainties2} as
\begin{equation} {
\begin{split}
    \begin{bmatrix}
\beta_k [n] - \sigma_k^2 & \left( \hat{\boldsymbol{h}}_k^\mathsf{H} [n]  +  \boldsymbol{e}_k^\mathsf{H} [n] \right) \boldsymbol{W}_{-{k}}[n]  \\
\boldsymbol{W}_{-{k}}^\mathsf{H} [n] \left( \hat{\boldsymbol{h}}_k [n]  +  \boldsymbol{e}_k [n] \right) & \boldsymbol{I}
\end{bmatrix} \\ \succeq \boldsymbol{0}, \quad\quad
    \forall \|\boldsymbol{e}_k[n]\| \leq \xi_k,\ \forall k \in \mathcal{K}, \ 
 \forall n \in \mathcal{N},
 \label{uncertainties4}
\end{split}}
\end{equation}where 
\begin{equation} {
    \boldsymbol{W}_{-k} [n] = \left [
        \boldsymbol{w}_1[n],\ldots,
        \boldsymbol{w}_{k-1}[n],
        \boldsymbol{w}_{k+1}[n], \ldots,
        \boldsymbol{w}_K[n]
    \right ].}
\end{equation}
Then, we introduce the slack variables $\boldsymbol{\varpi} [n] = \{ \varpi_1 [n],\ldots, \varpi_K [n]\}, \forall n$. Accordingly, we can equivalently transform \eqref{uncertainties4} into the following linear matrix inequality (LMI) \cite{9110587}
\begin{equation} {
    \begin{bmatrix}
\beta_k [n] - \sigma_k^2 - \varpi_k [n] & \hat{\boldsymbol{h}}_k^\mathsf{H} [n] \boldsymbol{W}_{-{k}} [n] & \boldsymbol{0}_{} \\
\boldsymbol{W}_{-{k}}^\mathsf{H} [n] \hat{\boldsymbol{h}}_k [n] & \boldsymbol{I}_{} & \xi_k \boldsymbol{W}_{-{k}}^\mathsf{H} [n] \\
\boldsymbol{0}_{} & \xi_k \boldsymbol{W}_{-{k}} [n]& \varpi_k [n] \boldsymbol{I}_{}
\end{bmatrix} \succeq \boldsymbol{0}.
\label{LMI2}}
\end{equation}

Then, \eqref{ineq1} can be rewritten as
\begin{equation} {
    {\alpha_k[n]} \geq {\beta_k[n]}\tilde{\gamma}_k[n],\ \forall k \in \mathcal{K}, \ 
 \forall n \in \mathcal{N}.
 \label{ineq2}}
\end{equation}
To transform \eqref{ineq2}, we introduce the following Lemma.
\begin{lemma} 
\label{lemmaup}
Given any two positive terms $x_{i}$ and $x_{j}$, the upper bound for $x_{i}x_{j}$ can be expressed as
\begin{equation} {
x_{i} x_{j} \leq \frac{1}{2} \left( \frac{x_{j}^{(l)}}{x_{i}^{(l)}} x_{i}^2 + \frac{x_{i}^{(l)}}{x_{j}^{(l)}} x_{j}^2 \right),}
\end{equation}
where $x_{i}^{(l)}$ and $x_{j}^{(l)}$ are the given local points in the $l$-th iteration of successive convex approximation (SCA).
\begin{proof}
	Please refer to \cite{7547360}.
\end{proof}
\end{lemma}
According to Lemma \ref{lemmaup}, we can obtain
\begin{equation} {
    {\alpha_k[n]} \geq \frac{1}{2} \left(
        \frac{\beta_k^{(l)}[n]}{\tilde{\gamma}_k^{(l)}[n]} \tilde{\gamma}_k^2 [n]+ \frac{\tilde{\gamma}_k^{(l)}[n]}{\beta_k^{(l)}[n]} \beta_k^2[n]
    \right).
     \label{ineq3}}
\end{equation}
As a result, problem \textbf{P2} can be reformulated as
\begin{subequations} { 
\begin{flalign}
 (\textbf{P2-1}):\ \max_{\boldsymbol{w},\boldsymbol{x}} \quad &  \frac{1}{N}\sum_{n= 1}^{N}\sum_{k= 1}^{K}  \tilde{r}_k[n]  \nonumber\\
 {\rm{s.t.}}  \quad & \eqref{p1a}\text{-}\eqref{p1c},\eqref{LMI1},\eqref{LMI2},\text{ and }\eqref{ineq3}.\nonumber
\end{flalign}}\end{subequations}

\subsection{Beamforming Design}
When given the antenna position $\boldsymbol{x}[n], \forall n$, problem \textbf{P2-1} can be simplified as
\begin{subequations} { 
\begin{flalign}
 (\textbf{P3}):\ \max_{\boldsymbol{w}} \quad &  \frac{1}{N} \sum_{n= 1}^{N}\sum_{k= 1}^{K}  \tilde{r}_k[n]  \nonumber\\
 {\rm{s.t.}}  \quad & \eqref{p1a},\eqref{LMI1},\eqref{LMI2},\text{ and }\eqref{ineq3}.\nonumber
\end{flalign}}\end{subequations}
Constraint \eqref{LMI1} is non-convex with respect to $\boldsymbol{w}_k[n]$. To tackle this, we approximiate $\boldsymbol{w}_k[n]\boldsymbol{w}_k^\mathsf{H}[n]$ by the first-order Taylor expansion of the SCA as
\begin{equation} {
\begin{split}
    \boldsymbol{w}_k[n]\boldsymbol{w}_k^\mathsf{H}[n] \succeq&\
    \boldsymbol{w}_k[n]{\left(\boldsymbol{w}_k^{(l)}\right)}^{\mathsf{H}} [n] \\&+
    \boldsymbol{w}_k^{(l)}[n]\boldsymbol{w}_k^\mathsf{H} -
    \boldsymbol{w}_k^{(l)}[n]{\left(\boldsymbol{w}_k^{(l)}\right)}^{\mathsf{H}} [n] \\
    &\triangleq \check{W}_k[n].
\end{split}}
\end{equation}
Substituting $\check{W}_k[n]$ into \eqref{LMI1}, we can obtain the following LMI
\begin{equation} {
    \begin{bmatrix} \check{W}_k[n] + \lambda_k[n]\boldsymbol{I} &  \check{W}_k[n] \hat{\boldsymbol{h}}_k^\mathsf{}[n] \\ \hat{\boldsymbol{h}}_k^\mathsf{H}[n]\check{W}_k[n] & \check{\kappa}_k[n]  \end{bmatrix}
        \succeq \boldsymbol{0},
    \label{lb1}}
\end{equation}
where
\begin{equation} {
    \check{\kappa}_k[n]  = \hat{\boldsymbol{h}}_k^\mathsf{H}[n]
    \check{W}_k[n][n]\hat{\boldsymbol{h}}_k^\mathsf{}[n] - \alpha_k[n] - \lambda_k[n]\xi_k^2.}
\end{equation}
Then, problem \textbf{P3} can be reformulated as
\begin{subequations} { 
\begin{flalign}
 (\textbf{P3-1}):\ \max_{\boldsymbol{w}} \quad &  \frac{1}{N}\sum_{n= 1}^{N}\sum_{k= 1}^{K}  \tilde{r}_k[n]  \nonumber\\
 {\rm{s.t.}}  \quad & \eqref{p1a},\eqref{LMI2},\eqref{ineq3},\text{ and }\eqref{lb1}.\nonumber
\end{flalign}}\end{subequations}
Problem \textbf{P3-1} is a standard convex optimization problem, which can be efficiently solved via the interior-point method (IPM) using convex optimization solvers such as CVX.

\subsection{Antenna Position Optimization}
When given the transmit beamforming $\boldsymbol{w}[n], \forall n$, problem \textbf{P2-1} can be reformulated as
\begin{subequations} { 
\begin{flalign}
 (\textbf{P4}):\ \max_{\boldsymbol{x}} \quad &  \frac{1}{N}\sum_{n= 1}^{N}\sum_{k= 1}^{K}  \tilde{r}_k[n]  \nonumber\\
 {\rm{s.t.}}  \quad & \eqref{p1b},\eqref{p1c},\eqref{LMI1},\eqref{LMI2},\text{ and }\eqref{ineq3}.\nonumber
\end{flalign}}\end{subequations}
Firstly, to handle the non-convex constraint \eqref{LMI1}, we rewrite
$\boldsymbol{w}_k^\mathsf{}[n]
    \boldsymbol{w}_k^\mathsf{H}[n] \hat{\boldsymbol{h}}_k^\mathsf{}[n]$ as $\frac{\sqrt{g_0}}{\hat{d}_k[n]} \boldsymbol{w}_k^\mathsf{}[n]
    \boldsymbol{w}_k^\mathsf{H}[n] \hat{\boldsymbol{a}}_k^\mathsf{}[n] $. Since $\boldsymbol{w}_k^\mathsf{}[n]
    \boldsymbol{w}_k^\mathsf{H}[n]$ is Hermitian, we can obtain
\begin{equation} {
\begin{split}
    \boldsymbol{w}_k^\mathsf{}[n]
    \boldsymbol{w}_k^\mathsf{H}[n] \hat{\boldsymbol{a}}_k^\mathsf{}[n]  &= \chi_k[n] \boldsymbol{u}_k[n]\boldsymbol{u}_k^\mathsf{H}[n] \hat{\boldsymbol{a}}_k^\mathsf{}[n]\\
    &= \chi_k[n]  \left (
        \sum_{p=1}^M \tilde{a}_{k,p}[n] 
    \right ) \boldsymbol{u}_k[n],
\end{split}}
\end{equation}
where $\chi_k[n]$ and $\boldsymbol{u}_k[n]$ are the singular value and corresponding singular vectors of $\boldsymbol{w}_k^\mathsf{}[n]
    \boldsymbol{w}_k^\mathsf{H}[n]$, respectively, and 
\begin{equation} {
    \tilde{a}_{k,p}[n] = \overline{[\boldsymbol{u}_k[n]]_p} \mathrm{e}^{\mathrm{j}\vartheta_k[n]x_p[n]},}
\end{equation}
where $\overline{[\boldsymbol{u}_k[n]]_p}$ denotes the complex conjugate of the $p$-th element of $\boldsymbol{u}_k[n]$,
and $\vartheta[n] =\frac{2\pi}{\lambda}\cos\hat{\theta}_k[n]$. Then, we have
\begin{equation} {
\begin{split}
    &\boldsymbol{w}_k^\mathsf{}[n]
    \boldsymbol{w}_k^\mathsf{H}[n] \hat{\boldsymbol{a}}_k^\mathsf{}[n] \approx
    \chi_k[n]  \left (
        \sum_{p=1}^M \tilde{a}_{k,p}^{(l)}[n]
    \right )\boldsymbol{u}_k[n]\\
    &+ \mathrm{j} \chi_k[n] \vartheta_k[n]  
    \left (
        \sum_{p=1}^M \tilde{a}_{k,p}^{(l)}[n]
        \left( x_p[n]- x_p^{(l)}[n]\right)
    \right ) \boldsymbol{u}_k[n]\\
    &\triangleq \boldsymbol{v}_k[n],
\end{split}}
\end{equation}
where
\begin{equation} {
    \tilde{a}_{k,p}^{(l)}[n] = \overline{[\boldsymbol{u}_k[n]]_p} \mathrm{e}^{\mathrm{j}\vartheta_k[n]x_p^{(l)}[n]}.}
\end{equation}
Similarly, we rewrite $\hat{\boldsymbol{h}}_k^\mathsf{H}[n]\boldsymbol{w}_k^\mathsf{}[n]
    \boldsymbol{w}_k^\mathsf{H}[n]$ as $\frac{\sqrt{g_0}}{\hat{d}_k[n]} \hat{\boldsymbol{a}}_k^\mathsf{H}[n]  \boldsymbol{w}_k^\mathsf{}[n]
    \boldsymbol{w}_k^\mathsf{H}[n] $. Then, we have
\begin{equation} {
\begin{split}
    \hat{\boldsymbol{a}}_k^\mathsf{H}[n]  \boldsymbol{w}_k^\mathsf{}[n]
    \boldsymbol{w}_k^\mathsf{H}[n] &= \chi_k[n] \left(
        \hat{\boldsymbol{a}}_k^\mathsf{H}[n] 
        \boldsymbol{u}_k[n]
    \right)
    \boldsymbol{u}_k^\mathsf{H}[n] \\
    &= \chi_k[n]  \left (
        \sum_{p=1}^M \bar{a}_{k,p}^{}[n]
    \right ) \boldsymbol{u}_k^\mathsf{H}[n],
\end{split}}
\end{equation}
where
\begin{equation} {
    \bar{a}_{k,p}^{}[n] = {[\boldsymbol{u}_k[n]]_p} \mathrm{e}^{-\mathrm{j}\vartheta_k[n]x_p[n]}.}
\end{equation}
Then, we can obtain the approximation of $\hat{\boldsymbol{a}}_k^\mathsf{H}[n]  \boldsymbol{w}_k^\mathsf{}[n]
    \boldsymbol{w}_k^\mathsf{H}[n] $ as
\begin{equation} {
\begin{split}
    &\hat{\boldsymbol{a}}_k^\mathsf{H}[n]  \boldsymbol{w}_k^\mathsf{}[n]
    \boldsymbol{w}_k^\mathsf{H}[n] \approx
    \chi_k[n]  \left (
        \sum_{p=1}^M \bar{a}_{k,p}^{(l)}[n]
    \right ) \boldsymbol{u}_k^\mathsf{H}[n]\\
    &-\mathrm{j} \chi_k[n] \vartheta_k[n]  
    \left (
        \sum_{p=1}^M \bar{a}_{k,p}^{(l)}[n]\left( x_p[n]- x_p^{(l)}[n]\right)    \right )
        \boldsymbol{u}_k^\mathsf{H}[n]\\
        &\triangleq \bar{\boldsymbol{v}}_k[n],
\end{split}}
\end{equation}
where
\begin{equation} {
    \bar{a}_{k,p}^{(l)}[n] = {[\boldsymbol{u}_k[n]]_p} \mathrm{e}^{-\mathrm{j}\vartheta_k[n]x_p^{(l)}[n]}.}
\end{equation}
Then, \eqref{kappa} can be rewritten as
\begin{equation} {
    \kappa_k[n] = \frac{g_0}{\hat{d}_k^2[n]
    } \left \vert
        \hat{\boldsymbol{a}}_k[n]
        \boldsymbol{w}_k[n]
    \right \vert^2  - \alpha_k[n] - \lambda_k[n]\xi_k^2.}
\end{equation}
Furthermore, $\left \vert
        \hat{\boldsymbol{a}}_k[n]
        \boldsymbol{w}_k[n]
    \right \vert^2 $ can be represented as
\begin{equation} {
\begin{split}
    \left| \hat{\boldsymbol{{a}}}_{k}^{\mathsf{H}}[n]\boldsymbol{w}_{k}[n] \right|^{2} 
& = \left| \sum_{p=1}^{M} [w_{k}^{}[n] ]_p e^{{\mathrm{j}} \vartheta_{k}[n] x_{p}[n]  }  \right|^{2} 
\\
& = \sum_{p=1}^{M} \sum_{q=1}^{M} 
\left| [w_{k}^{}[n] ]_p [w_{k}^{}[n] ]_q \right|
\cos \Theta_{k,p,q}[n],
\end{split}}
\end{equation}
where $[w_{k}[n]]_p = |\boldsymbol{w}_{k}[n]| {\mathrm{e}}^{{\mathrm{j}} \angle [w_{k}[n]]_p}$ is the $p$-th element of the beamforming vector $\boldsymbol{w}_{i}[n]$ and $ \Theta_{k,p,q}[n] = \vartheta_{k}[n] (x_{p}[n]-x_{q}[n]) -( \angle [w_{k}[n]]_p - \angle [w_{k}[n]]_q) $. Then, we can obtain the approximation of $\left \vert
        \hat{\boldsymbol{a}}_k[n]
        \boldsymbol{w}_k[n]
    \right \vert^2 $ by the following lemma.
\begin{lemma}
For any $a \in \mathbb{R}$, we can obtain the lower bound of $\cos (a)$ by second-order Taylor expansion as
\begin{equation} {
\begin{split} 
&\cos (a) \approx \cos (a_0) - \sin (a_0) (a -a_0) - \frac{1}{2} \cos (a_0) (a -a_0)^{2}  \\
&\ \ \geq \cos (a_0) - \sin (a_0) (a -a_0) - \frac{1}{2} (a -a_0)^{2} = \check{f}\left(a \mid a_0 \right).
\end{split}}
\end{equation}
\label{lemma2nd}
\end{lemma}
Following Lemma \ref{lemma2nd}, we have
\begin{equation} {
    \left \vert
        \hat{\boldsymbol{a}}_k[n]
        \boldsymbol{w}_k[n]
    \right \vert^2 \geq \frac{1}{2} \boldsymbol{x}^{\mathsf{T}}[n] \boldsymbol{X}_{k}[n] \boldsymbol{x}[n] +\boldsymbol{y}_{k}^{\mathsf{T}}[n] \boldsymbol{x}[n]+z_{k}[n],}
\end{equation}
where 
\begin{equation} {
\boldsymbol{X}_{k}[n] =  -2 \vartheta_{k}^{2}[n]\left(\eta[n] \operatorname{diag}(\tilde{\boldsymbol{w}}_{k}[n])-\tilde{\boldsymbol{w}}_{k}[n] \tilde{\boldsymbol{w}}_{k}[n]^{\mathsf{T}}\right),}
\end{equation}
\begin{equation} {
    \begin{split}
 [\boldsymbol{y}_{k}[n]]_{p} = &\ 2 \vartheta_{k}^{2}[n] \sum_{q=1}^{M}\left|[\boldsymbol{w}_{k}[n]]_{p} [\boldsymbol{w}_{k}[n]]_{q}\right|\left(x_{p}^{(l)}[n]-x_{q}^{(l)}[n]\right) \\
& -2 \vartheta_{k}[n] \sum_{q=1}^{M}\left|[\boldsymbol{w}_{k}[n]]_{p} [\boldsymbol{w}_{k}[n]]_{q}\right| \sin \left( \Theta_{k,p,q}^{(l)}[n] \right),       
    \end{split}}
\end{equation}
\begin{equation} {
    \begin{aligned}
&z_{k}[n] =  \sum_{p=1}^{M} \sum_{q=1}^{M}\left|[\boldsymbol{w}_{k}[n]]_{p} [\boldsymbol{w}_{k}[n]]_{q} \right| \cos \left( \Theta_{k,p,q}^{(l)}[n] \right) \\
& +\vartheta_{k}[n] \sum_{p=1}^{M} \sum_{q=1}^{M}\left|[\boldsymbol{w}_{k}[n]]_{p} [\boldsymbol{w}_{k}[n]]_{q}\right| \sin \left( 
 \Theta_{k,p,q}^{(l)}[n] \right) \times \\
 & \left(x_{p}^{(l)}[n]-x_{q}^{(l)}[n]\right) 
 -\frac{1}{2} \vartheta_{k}^{2}[n] \sum_{p=1}^{M} \sum_{q=1}^{M}\left|[\boldsymbol{w}_{k}[n]]_{p} [\boldsymbol{w}_{k}[n]]_{q}\right| \times \\
& \left(x_{p}^{(l)}[n]-x_{q}^{(l)}[n]\right)^{2},  
    \end{aligned}}
\end{equation}
with $\eta[n] = \sum_{p=1}^{M} | [\boldsymbol{w}_{k}[n]]_{p}| $ and $\Tilde{\boldsymbol{w}}_{k}[n] = [ |[\boldsymbol{w}_{k}[n]]_{1}|, \ldots, |[\boldsymbol{w}_{k}[n]]_{M}| ]^{\mathsf{T}} $.
The new slack variables $\varrho_k [n], \forall k,n$ satisfy
\begin{equation} {
    \frac{1}{2} \boldsymbol{x}^{\mathsf{T}}[n] \boldsymbol{X}_{k}[n] \boldsymbol{x}[n] +\boldsymbol{y}_{k}^{\mathsf{T}}[n] \boldsymbol{x}[n]+z_{k}[n] \geq \varrho_k [n].}
\end{equation}
Then, \eqref{kappa} can be approximated as
\begin{equation} {
    \kappa_k[n] \geq  \kappa_k[n] = \frac{g_0}{\hat{d}_k^2[n]
    } \varrho_k [n]  - \alpha_k[n] - \lambda_k[n]\xi_k^2 \triangleq \bar{\kappa}_k [n].}
\end{equation}
As a result, constraint \eqref{LMI1} can be reformulated as
\begin{equation} {
    \begin{bmatrix} \boldsymbol{w}_k^\mathsf{}[n]
    \boldsymbol{w}_k^\mathsf{H}[n] + \lambda_k[n]\boldsymbol{I} &  \frac{\sqrt{g_0}}{\hat{d}_k[n]} \boldsymbol{v}_k[n] \\ \frac{\sqrt{g_0}}{\hat{d}_k[n]} \bar{\boldsymbol{v}}_k[n] & \bar{\kappa}_k[n]  \end{bmatrix}
        \succeq \boldsymbol{0}.
        \label{LMI3}}
\end{equation}

Then, we rewrite $\hat{\boldsymbol{h}}_k^\mathsf{H} [n] \boldsymbol{W}_{-{k}} [n]$ as $\frac{\sqrt{g_0}}{\hat{d}_k[n]} \hat{\boldsymbol{a}}_k^\mathsf{H} [n] \boldsymbol{W}_{-{k}} [n]$, where $\hat{\boldsymbol{a}}_k^\mathsf{H} [n] \boldsymbol{W}_{-{k}} [n]$ can be approximated as
\begin{equation} {
\begin{split}
    &\hat{\boldsymbol{a}}_k^\mathsf{H} [n] \boldsymbol{W}_{-{k}} [n] = 
    \left [
     \hat{\boldsymbol{a}}_k^\mathsf{H} [n] \boldsymbol{w}_1 [n],\ldots,
    \hat{\boldsymbol{a}}_k^\mathsf{H} [n] \boldsymbol{w}_{k-1} [n],\right.\\&\quad\quad\quad\quad\quad\quad\quad\ \left.
    \hat{\boldsymbol{a}}_k^\mathsf{H} [n] \boldsymbol{w}_{k+1} [n],\ldots,
    \hat{\boldsymbol{a}}_k^\mathsf{H} [n] \boldsymbol{w}_K [n]
    \right ]\\&
    \approx
    \left [ 
        \tilde{A}_{1}^k[n], \ldots, \tilde{A}_{{k-1}}^k[n], \tilde{A}_{{k+1}}^k[n], \ldots, \tilde{A}_{K}^k[n]
    \right] \triangleq \tilde{\boldsymbol{A}}_k[n],
\end{split}}
\end{equation}
where
\begin{equation} {
\begin{split}
    &\tilde{A}_{i}^k [n] = \sum_{p=1}^M [\boldsymbol{w}_i[n]]_p\mathrm{e}^{
    -\mathrm{j}\vartheta_k[n]x_p^{(l)}[n]
    } 
    \\&\ \ \ - \mathrm{j}\vartheta_k[n] \left (
        \sum_{p=1}^M [\boldsymbol{w}_i[n]]_p\mathrm{e}^{
    -\mathrm{j}\vartheta_k[n]x_p^{(l)}[n]} \left( x_p[n] -x_p^{(l)}[n]
    \right)
    \right ),\\
    &\quad\quad\quad\quad\quad\quad\quad\quad\quad\quad\quad\quad\quad\quad\ \ \forall i \in \mathcal{M},\ \text{}\ i \neq k.
\end{split}}
\end{equation}
Similarly, $\boldsymbol{W}_{-{k}}^\mathsf{H} [n] \hat{\boldsymbol{h}}_k [n] $ can be expressed as $\frac{\sqrt{g_0}}{\hat{d}_k[n]} \boldsymbol{W}_{-{k}}^\mathsf{H} [n] \hat{\boldsymbol{a}}_k [n]$, where $\boldsymbol{W}_{-{k}}^\mathsf{H} [n] \hat{\boldsymbol{a}}_k [n]$ can be approximated as
\begin{equation} {
\begin{split}
    &\boldsymbol{W}_{-{k}}^\mathsf{H} [n] \hat{\boldsymbol{a}}_k [n] = 
    \left [
      \boldsymbol{w}_1^\mathsf{H} [n]\hat{\boldsymbol{a}}_k [n],\ldots,
    \boldsymbol{w}_{k-1}^\mathsf{H} [n]\hat{\boldsymbol{a}}_k [n] ,\right.\\&\quad\quad\quad\quad\quad\quad\quad\ \left.
     \boldsymbol{w}_{k+1}^\mathsf{H}  [n]\hat{\boldsymbol{a}}_k[n],\ldots,
     \boldsymbol{w}_K^\mathsf{H} [n]\hat{\boldsymbol{a}}_k [n]
    \right ]^\mathsf{T}\\&
    \approx
    \left [ 
        \tilde{B}_{1}^k[n], \ldots, \tilde{B}_{{k-1}}^k[n], \tilde{B}_{{k+1}}^k[n], \ldots, \tilde{B}_{K}^k[n]
    \right]^\mathsf{T} \triangleq \tilde{\boldsymbol{B}}_k[n],
\end{split}}
\end{equation}
where
\begin{equation} {
\begin{split}
    &\tilde{B}_{i}^k [n] = \sum_{p=1}^M \overline{[\boldsymbol{w}_i[n]]_p}\mathrm{e}^{
    \mathrm{j}\vartheta_k[n]x_p^{(l)}[n]
    } 
    \\&\ \ \ + \mathrm{j}\vartheta_k[n] \left (
        \sum_{p=1}^M \overline{[\boldsymbol{w}_i[n]]_p}\mathrm{e}^{
    \mathrm{j}\vartheta_k[n]x_p^{(l)}[n]} \left( x_p[n] -x_p^{(l)}[n]
    \right)
    \right ),\\
    &\quad\quad\quad\quad\quad\quad\quad\quad\quad\quad\quad\quad\quad\quad\ \ \forall i \in \mathcal{M},\ \text{}\ i \neq k.
\end{split}}
\end{equation}
Accordingly, constraint \eqref{LMI2} can be reformulated as
\begin{equation} {
    \begin{bmatrix}
\beta_k [n] - \sigma_k^2 - \varpi_k [n] & \frac{\sqrt{g_0}}{\hat{d}_k[n]}\tilde{\boldsymbol{A}}_k[n]  & \boldsymbol{0}_{} \\
\frac{\sqrt{g_0}}{\hat{d}_k[n]}\tilde{\boldsymbol{B}}_k[n] & \boldsymbol{I}_{} & \xi_k \boldsymbol{W}_{-{k}}^\mathsf{H} [n] \\
\boldsymbol{0}_{} & \xi_k \boldsymbol{W}_{-{k}} [n]& \varpi_k [n] \boldsymbol{I}_{}
\end{bmatrix} \succeq \boldsymbol{0}.
\label{LMI4}}
\end{equation}
To this end, problem \textbf{P4} can be reformulated as
\begin{subequations} { 
\begin{flalign}
 (\textbf{P4-1}):\ \max_{\boldsymbol{x}} \quad &  \frac{1}{N}\sum_{n= 1}^{N}\sum_{k= 1}^{K}  \tilde{r}_k[n]  \nonumber\\
 {\rm{s.t.}}  \quad & \eqref{p1b},\eqref{p1c}, \eqref{ineq3},\eqref{LMI3},\text{ and }\eqref{LMI4}.\nonumber
\end{flalign}}\end{subequations}
Problem \textbf{P4-1} is now reformulated as a standard convex optimization problem, which can be efficiently solved using convex optimization solvers such as CVX.

\subsection{Algorithm Analysis}
{
We summarize the whole procedure of the AO algorithm in Algorithm \ref{alg:AO}.
The algorithm is initialized in Line 1 with feasible solutions of the transmit beamforming and antenna positions, and the iteration index is initialized as $0$.
From Lines 2 to 9, the AO algorithm proceeds through repeated iterations.
Specifically, Line 3 and 4 optimize the transmit beamforming $\boldsymbol{w}$ via IPM while fixing antenna positions; Line 5 and 6 optimize the antenna positions $\boldsymbol{x}$ via IPM while fixing the beamforming; Line 7 updates the objective value; Line 8 increments the iteration index; and Line 9 checks convergence, i.e., stopping if the objective change is below the threshold or the max iteration number is reached.
Finally, the algorithm outputs the converged optimized solution.
Moreover, the following Proposition characterizes the convergence behavior of the proposed AO-based algorithm.
}
\begin{algorithm}[t]
\caption{{AO-based Robust Optimization Algorithm for MA-Enabled CVNs.}}
\label{alg:AO} 
\begin{algorithmic}[1]
\REQUIRE
{Locations of BS and VUs, estimated CSI $\hat{\boldsymbol{h}}_k[n]$, uncertainty radii $\xi_k$, iteration index $ l $, maximum iteration number $l_{\rm{max}}$, accuracy threshold $ \varepsilon > 0 $.
\STATE
\textbf{Initialize:} Feasible solution $\{\boldsymbol{w}^{(0)}[0]\}$ and $\{ \boldsymbol{x}_{}^{(0)}[n] \}$, $ l = 1 $.}
\STATE
\textbf{Repeat:}
\STATE \textbf{\quad\ Step 1:} beamforming optimization (\textbf{P3-1})
\STATE \textbf{\quad\quad\ \ }Solve problem \textbf{P3-1}
via IPM to obtain $\boldsymbol{w}^{(l+1)}$;
\STATE \textbf{\quad\ Step 2:} antenna position optimization (\textbf{P4-1})
\STATE \textbf{\quad\quad\ \ }Solve problem \textbf{P4-1}
via IPM to obtain $\boldsymbol{x}^{(l+1)}$;
\STATE Update the objective function value of problem \textbf{P2-1} $\Phi^{(l+1)}$ according to the obtained $\boldsymbol{w}^{(l+1)}$ and $\boldsymbol{x}^{(l+1)}$.
\STATE Update $ l \leftarrow l + 1 $.
\STATE
\textbf{Until:} $\left| \Phi^{(l)} - \Phi^{(l-1)} \right| \le \varepsilon$ or $l > l_{\max}$.
\ENSURE
{The optimized solution for the transmit beamforming $\boldsymbol{w}^{(l)}$, and the antenna position $\boldsymbol{x}^{(l)}$.}
\end{algorithmic} 
\end{algorithm}

{
\begin{proposition}
    (Convergence Analysis): The proposed AO algorithm is guaranteed to converge to a stationary solution.
    \label{propno2}
\end{proposition}
\begin{proof}
Please refer to Appendix \ref{apppro2}.
\end{proof}
}

{
The overall computational complexity of Algorithm~1 is dominated by the per-iteration cost of solving the convex subproblems \textbf{P3-1} and \textbf{P4-1} via IPM. Specifically, the beamforming optimization problem \textbf{P3-1} involves $\mathcal{O}(KNM)$ complex variables and the LMI constraints, with the complexity of $\mathcal{O}\big( (KNM)^{3.5} \log(1/\varepsilon) \big)$. Accordingly, the antenna position optimization problem \textbf{P4-1} has complexity of $\mathcal{O}\big( (KN)^{3.5} \log(1/\varepsilon) \big)$. With at most $l_{\max}$ outer iterations, the total complexity is $\mathcal{O}\big( l_{\max} \cdot (KNM)^{3.5} \log(1/\varepsilon) \big)$.
}

\section{Numerical Results}

In this section, we assess the performance of the proposed scheme for the MA-enabled CVN. The system comprises a BS positioned at the center of a squared coverage area, specifically at coordinates $[250~\mathrm{m}, 250~\mathrm{m}]$, while two VUs move according to a preset trajectory.
Unless otherwise specified, the simulation parameters are set as follows: the total number of MAs is $M = 4$, the service period consists of $N = 6$ time slots, the BS is deployed at an altitude of $H = 12~\mathrm{m}$, and the reference noise power at each VU is $\sigma_k^2 = -80~\mathrm{dBm}$. The maximum allowable transmit power of the BS is $P_{\max} = 34~\mathrm{dBm}$. The movable range for the MAs is constrained to $[0, L
_{\mathrm{}}] = [0, 6\lambda]$, and the minimum separation between any adjacent antennas is $d_{\mathrm{min}} = 0.3\lambda$.
To demonstrate the performance gain of the proposed scheme, we compare it with the following three benchmark schemes:
\begin{itemize}
    \item {\textbf{Upper bound scheme}: In this scheme, it is assumed that the system has perfect knowledge of the CSI.}
    \item \textbf{FPA scheme}: In this scheme, the transmit beamforming of the BS is optimized by solving problem \textbf{P3-1}, while the positions of the antennas are fixed.
    \item \textbf{Fixed beamforming (FB) scheme}: In this scheme, the antenna positions are optimized by solving problem \textbf{P4-1}, whereas the transmit beamforming of the BS is fixed (e.g., designed based on a predefined strategy or the initial solution).
\end{itemize}

\begin{figure}[t]
		\centering
		\includegraphics[width=0.9\linewidth]{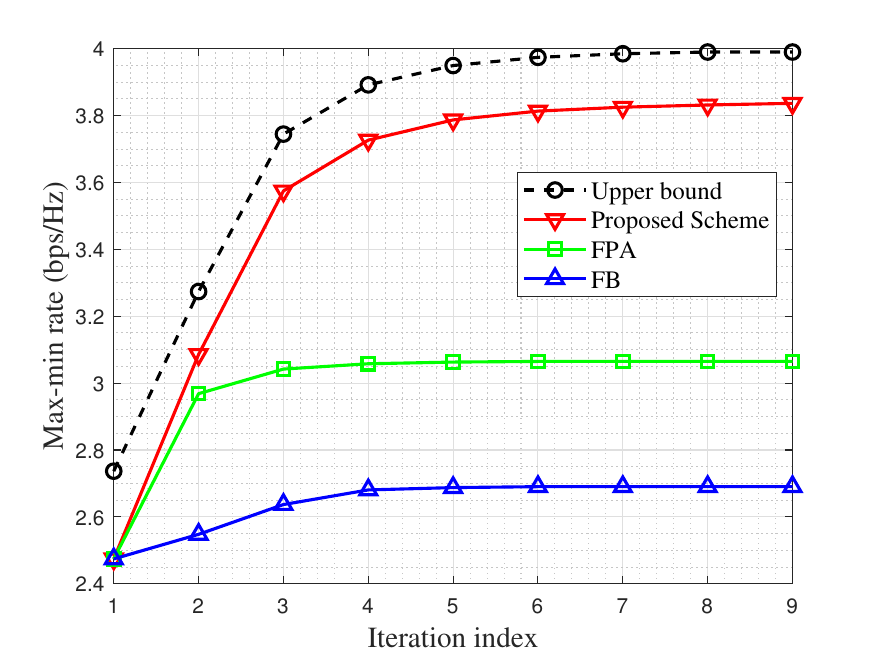}
		\captionsetup{font=small}
            \caption{\centering{Convergence behavior of different schemes.
            }}
		\label{fig: Convergence}
\end{figure}
\begin{figure}[t]
		\centering
		\includegraphics[width=0.9\linewidth]{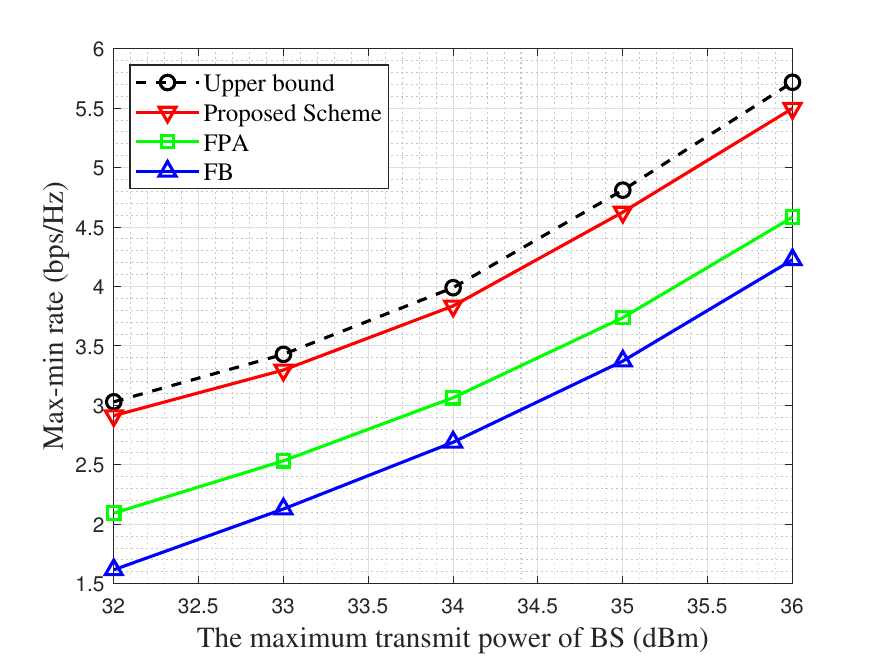}
            \captionsetup{font=small}
		\caption{\centering{The Max-min data rate versus the maximum transmit power.
            }}
		\label{fig: Transmit Power}
\end{figure}

Fig.~\ref{fig: Convergence} illustrates the convergence behavior of all schemes. As we can observe, the proposed scheme attains the highest max-min data rate compared to the FPA scheme and the FB scheme, confirming the effectiveness of jointly optimizing the transmit beamforming and the antenna positions. Compared with the upper bound scheme, the proposed scheme achieves performance close to it, indicating that under the system constraints, the proposed scheme effectively exploits the synergistic gains of antenna placement and transmit beamforming.
Moreover, the proposed scheme surpasses both the FPA scheme and the FB scheme. In the FPA scheme, the fixed antenna positions prevent exploitation of positional flexibility for performance gain. In the FB scheme, although positions are adjustable, beamforming is designed independently of channel conditions, limiting its ability to mitigate interference or enhance the desired signals.

\begin{figure}[t]
		\centering
		\includegraphics[width=0.9\linewidth]{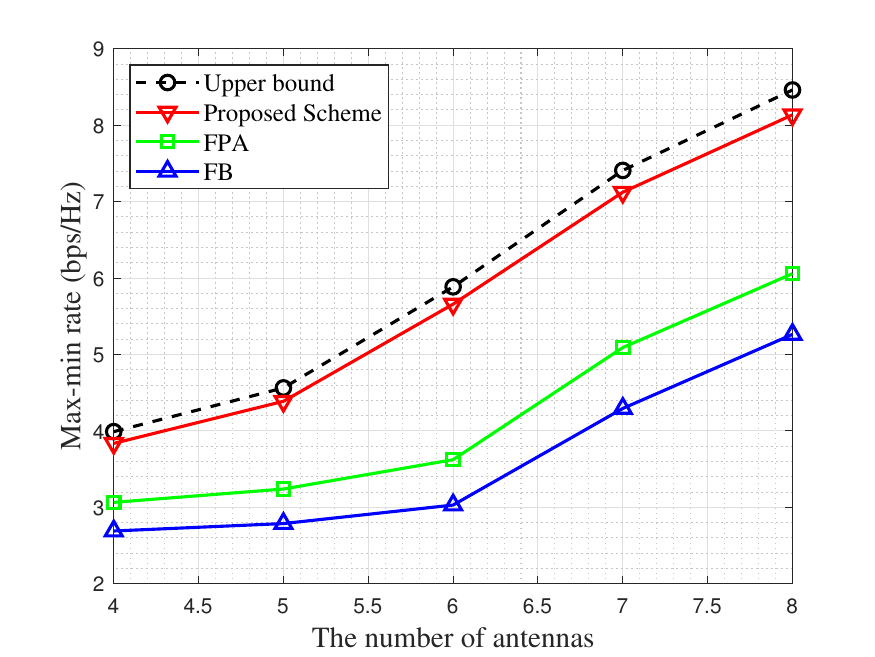}
            \captionsetup{font=small}
		\caption{\centering{The Max-min data rate versus the number of antennas.
            }}
		\label{fig: number_antenna}
\end{figure}
\begin{figure}[t]
		\centering
		\includegraphics[width=0.9\linewidth]{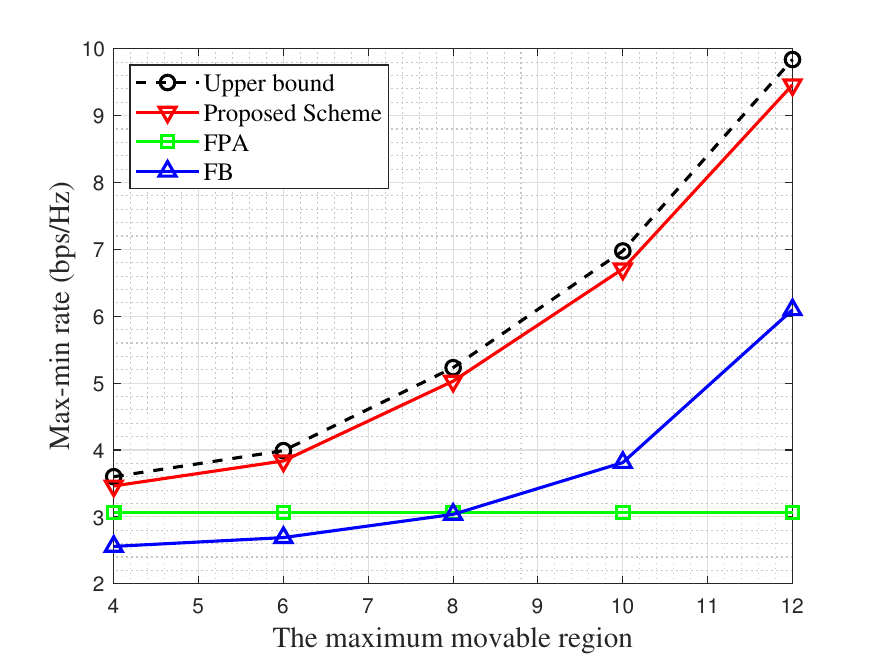}
            \captionsetup{font=small}
		\caption{\centering{The Max-min data rate versus the maximum movable region.
            }}
		\label{fig: range_antenna}
\end{figure}

Fig.~\ref{fig: Transmit Power} shows the variation of the max-min data rate with the maximum transmit power of the BS. It can be seen that as the transmit power increases, the performance of all schemes gradually improves. The reason is that, higher transmit power of the BS can enhance the received signal strength, thereby increasing the max-min data rate of VUs. Due to the joint optimization of the transmit beamforming and the antenna positions, the proposed scheme can improve the desired signal power and effectively suppress multi-user interference. Therefore, it is always superior to the FPA and FB schemes over the entire power range, and is clost to the upper bound scheme.

{
Fig.~\ref{fig: number_antenna} illustrates the variation of the max–min data rate with the number of antennas of the BS. As the number of antennas increases, the performance of all schemes improves, since more antennas provide greater spatial DoFs, enhancing both signal strengthening and interference suppression capabilities. However, the proposed scheme and the upper bound scheme exhibit more pronounced performance gains. This is primarily attributed to the joint optimization of antenna positions and transmit beamforming, which allows the additional antennas to be utilized more efficiently, resulting in stronger beam focusing and better interference management.
}

{
Fig.~\ref{fig: range_antenna} shows the variation of the max-min data rate with the size of the movable region of the MA. It can be observed that the max-min data rates of the proposed scheme, the upper bound scheme, and the FB scheme all gradually increase as the movable region of the antenna expands. This is mainly because a larger movable range provides more spatial DoFs, enabling the BS to achieve more precise and flexible beamforming while effectively suppressing energy leakage to undesired regions. In addition, when the maximum movable region exceeds $8\lambda$, the performance of the FB scheme surpasses that of the FPA scheme, which indicates that under a larger moving range, the flexibility of the antenna position can significantly improve system performance by the spatial DoFs even if beamforming is not optimized with the CSI.
}

\begin{figure}[t]
		\centering
		\includegraphics[width=0.9\linewidth]{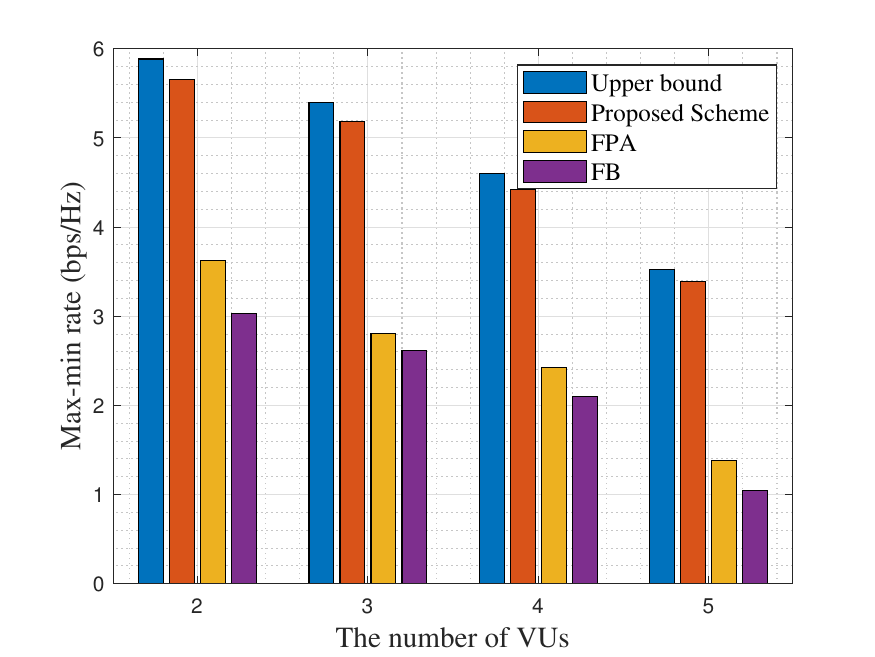}
            \captionsetup{font=small}
		\caption{\centering{The Max-min data rate versus the number of VUs.
            }}
		\label{fig: number_VU}
\end{figure}
{
Fig.~\ref{fig: number_VU} presents the variation of the max-min data rate with the number of VUs. It can be observed that under different numbers of VUs, the performance of the proposed scheme is always the closest to the upper bound scheme and better than the FPA scheme and the FB scheme, indicating that the proposed scheme is more robust against changes in the number of VUs. As the number of VUs increases, the max-min data rate of all schemes decreases. The reason is that, more VUs lead to more severe inter-VU interference, and at the same time, the spatial DoFs and transmit power resources that can be allocated to each VU are reduced, thereby limiting the performance improvement of the system.
}

\section{Conclusion}
This paper tackles the issue of imperfect CSI in high-mobility CVNs by integrating MA technology. We propose a robust joint design of the transmit beamforming and the antenna positions to maximize the worst-case VU rate, accounting for bounded channel estimation uncertainty. Leveraging an AO algorithm, the original non-convex problem is decomposed into tractable convex subproblems, which are solved via standard convex solvers. Simulation results demonstrate that the proposed scheme significantly outperforms the FPA scheme and the FB scheme in terms of worst-case max-min rate, validating the advantage of spatial adaptability enabled by MAs. These findings highlight the promise of the MA-enabled BSs in improving communication robustness and spectral efficiency under dynamic CVNs.

\begin{appendices}

\section{Proof of Proposition \ref{propno1}}
\begin{proof}
    \label{apppro1}

{
Let $d_k [n] = \sqrt{\Vert \boldsymbol{b}-{\boldsymbol{q}}_k[n]\Vert^2+H^2}$ and $\hat{d}_k [n] = \sqrt{\Vert \boldsymbol{b}-\hat{\boldsymbol{q}}_k[n]\Vert^2+H^2}$ denote the true distance and estimated distance between the BS and the $k$-th VU at time slot $n$. Then, the estimation error can be rewritten as
    \begin{equation} {
        \boldsymbol{e}_k[n] = 
        \frac{\sqrt{g_0}}{d_k[n]}\boldsymbol{a}_k(\boldsymbol{x}[n], {\theta}_k[n]) - \frac{\sqrt{g_0}}{\hat{d}_k[n]}\hat{\boldsymbol{a}}_k(\boldsymbol{x}[n], \hat{\theta}_k[n]).
        \label{prop1}}
    \end{equation}
    Then, we introduce an intermediate term $\frac{\sqrt{g_0}}{\hat{d}_k[n]}\boldsymbol{a}_k(\boldsymbol{x}[n], {\theta}_k[n]) $. According to the triangle inequality, \eqref{prop1} can be expressed as
    \begin{equation} {
    \begin{split}
        \boldsymbol{e}_k[n] &= \frac{\sqrt{g_0}}{d_k[n]}\boldsymbol{a}_k(\boldsymbol{x}[n], {\theta}_k[n]) - \frac{\sqrt{g_0}}{\hat{d}_k[n]}\boldsymbol{a}_k(\boldsymbol{x}[n], {\theta}_k[n]) \\&\quad+ \frac{\sqrt{g_0}}{\hat{d}_k[n]}\boldsymbol{a}_k(\boldsymbol{x}[n], {\theta}_k[n]) - \frac{\sqrt{g_0}}{\hat{d}_k[n]}\hat{\boldsymbol{a}}_k(\boldsymbol{x}[n], \hat{\theta}_k[n])\\
        &\leq \left \Vert \left( \frac{\sqrt{g_0}}{{d_k[n]}} - \frac{\sqrt{g_0}}{\hat{d}_k[n]} \right) \boldsymbol{a}_k(\boldsymbol{x}[n], {\theta}_k[n]) \right \Vert\\
        &\quad+ \left \Vert \frac{\sqrt{g_0}}{\hat{d}_k[n]} \left( \boldsymbol{a}_k(\boldsymbol{x}[n], {\theta}_k[n]) - \hat{\boldsymbol{a}}_k(\boldsymbol{x}[n], \hat{\theta}_k[n]) \right) \right \Vert\\
        &=\sqrt{M} {\alpha}_k[n]  + 
        \frac{\sqrt{g_0}}{\hat{d}_k[n]} \beta_k[n],
    \end{split}
    \label{propek1}}
    \end{equation}
    where
    \begin{equation} {
        {\alpha}_k[n] = \sqrt{g_0} \left \vert \frac{1}{{d_k[n]}} - \frac{1}{\hat{d}_k[n]} \right \vert,
        \label{propalpha}}
    \end{equation}
    \begin{equation} {
        \beta_k[n] = \left \Vert \boldsymbol{a}_k(\boldsymbol{x}[n], {\theta}_k[n]) - \hat{\boldsymbol{a}}_k(\boldsymbol{x}[n], \hat{\theta}_k[n]) \right \Vert.}
    \end{equation}
    According to the Lipschitz continuity of the distance, we have
    \begin{equation} {
        |d_k[n] - \hat{d}_k[n]| \leq \Vert \boldsymbol{q}_k[n] - \hat{\boldsymbol{q}}_k[n] \Vert \leq r_k.\label{Lipschitzc}}
    \end{equation}
    Combining \eqref{propalpha} and \eqref{Lipschitzc}, we have
    \begin{equation} {
        \alpha_k[n] = \sqrt{g_0} \frac{\vert d_k[n]-\hat{d}_k[n]\vert}{d_k[n]\hat{d}_k[n]} \leq \frac{\sqrt{g_0} r_k}{d_k[n]\hat{d}_k[n]}.
        \label{propal2}}
    \end{equation}
    For the position of the $m$-th MA $x_m[n] \in [0, L]$, the phase error of the steering vector is bounded by the maximum phase difference $\frac{2\pi L }{\lambda}\cos\theta$.
    Using $\left\vert\cos\theta_k [n] - \cos\hat{\theta}_k[n]\right\vert \leq \frac{Hr_k}{d_k[n]\hat{d}_k[n]}$, 
    we can obtain
    \begin{equation} {\begin{split}
        \beta_k[n] &\leq \sqrt{M} \frac{2\pi L }{\lambda} \left\vert\cos\theta_k [n] - \cos\hat{\theta}_k[n]\right\vert \\
        &\leq \sqrt{M} \frac{2\pi L H r_k }{\lambda d_k[n]\hat{d}_k[n]}.
    \end{split}}
    \label{propbeta}
    \end{equation}
    Combining \eqref{propal2} and \eqref{propbeta} into \eqref{propek1}, we can obtain \eqref{propek}.
}
\end{proof}

\section{Proof of Proposition \ref{propno2}}
\begin{proof}
    \label{apppro2}

{
        Firstly, Proposition \ref{propno1} shows that the channel estimation error is upper bounded. Meanwhile, the transmit power of the BS is also upper bounded. Hence, the objective function of \textbf{P2-1} is upper bounded.
        Besides, the objective function is continuous with respect to $\boldsymbol{w}$ and $\boldsymbol{x}$, and the feasible set is closed and bounded. Therefore, the objective function attains its maximum over the feasible set.

        Let $\boldsymbol{w}^{(l)}$ and $\boldsymbol{x}^{(l)}$ denote the solutions obtained by Algorithm 1 in the $l$-th iteration, and $\Psi \{\boldsymbol{w}^{(l)}, \boldsymbol{x}^{(l)} \}$ denote the value of the objective function of \textbf{P2-1} obtained by the $l$-th iteration. Then, given the antenna position $\boldsymbol{x}^{(l)}$, we optimize the transmit beamforming by solving \textbf{P3-1} to obtain $\boldsymbol{w}^{(l+1)}$. Then, we have
        \begin{equation}
            \Psi \{\boldsymbol{w}^{(l+1)}, \boldsymbol{x}^{(l)} \} \geq \Psi \{\boldsymbol{w}^{(l)}, \boldsymbol{x}^{(l)} \}.
        \end{equation}
        Accordingly, when given the optimized transmit beamforming $\boldsymbol{w}^{(l+1)}$, we optimize the antenna position by solving \textbf{P4-1} to obtain $\boldsymbol{x}^{(l+1)}$. Then, we have
        \begin{equation}
            \Psi \{\boldsymbol{w}^{(l+1)}, \boldsymbol{x}^{(l+1)} \} \geq \Psi \{\boldsymbol{w}^{(l+1)}, \boldsymbol{x}^{(l)} \}.
        \end{equation}
        Therefore, after each iteration of the AO algorithm, the value of the objective function is monotonically non-decreasing. In summary, Algorithm 1 can be guaranteed to converge.
}
\end{proof}

\end{appendices}

\bibliographystyle{IEEEtran}
\bibliography{myref}

@ARTICLE{9509294,
  author={Nguyen, Dinh C. and Ding, Ming and Pathirana, Pubudu N. and Seneviratne, Aruna and Li, Jun and Niyato, Dusit and Dobre, Octavia and Poor, H. Vincent},
  journal={IEEE Internet Things J.}, 
  title={{6G} Internet of Things: A Comprehensive Survey}, 
  year={Jan. 2022},
  volume={9},
  number={1},
  pages={359-383},}

@ARTICLE{9779322,
  author={Noor-A-Rahim, Md. and Liu, Zilong and Lee, Haeyoung and Khyam, Mohammad Omar and He, Jianhua and Pesch, Dirk and Moessner, Klaus and Saad, Walid and Poor, H. Vincent},
  journal={Proc. IEEE}, 
  title={{6G} for Vehicle-to-Everything ({V2X}) Communications: Enabling Technologies, Challenges, and Opportunities}, 
  year={Jun. 2022},
  volume={110},
  number={6},
  pages={712-734},}

@ARTICLE{9146378,
  author={Haydari, Ammar and Yılmaz, Yasin},
  journal={IEEE Trans. Intell. Transp. Syst.}, 
  title={Deep Reinforcement Learning for Intelligent Transportation Systems: A Survey}, 
  year={Jan. 2022},
  volume={23},
  number={1},
  pages={11-32},}

@ARTICLE{10133894,
  author={Gong, Taiyuan and Zhu, Li and Yu, F. Richard and Tang, Tao},
  journal={IEEE Trans. Intell. Transp. Syst.}, 
  title={Edge Intelligence in Intelligent Transportation Systems: A Survey}, 
  year={Sep. 2023},
  volume={24},
  number={9},
  pages={8919-8944},
}

@ARTICLE{10287319,
  author={Chi, Hao Ran and Radwan, Ayman and Zhang, Chunjiong and Taha, Abd-Elhamid M.},
  journal={IEEE Commun. Standards Mag.}, 
  title={Managing Energy-Experience Trade-Off with {AI} Towards {6G} Vehicular Networks}, 
  year={Sep. 2023},
  volume={7},
  number={3},
  pages={24-31},}

@ARTICLE{10874187,
  author={Ke, Yan-Hui and Zhou, Jian and Chen, Jian-Xin},
  journal={IEEE Trans. Veh. Technol.}, 
  title={Dual-Band Dual-Polarized Omnidirectional {MIMO} Antenna With Ultra-Low Profile for Vehicular Communications}, 
  year={Jun. 2025},
  volume={74},
  number={6},
  pages={9240-9251},}

@ARTICLE{9264694,
  author={Wong, Kai-Kit and Shojaeifard, Arman and Tong, Kin-Fai and Zhang, Yangyang},
  journal={IEEE Trans. Wireless Commun.}, 
  title={Fluid Antenna Systems}, 
  year={Mar. 2021},
  volume={20},
  number={3},
  pages={1950-1962},}

@article{liu2025movable,
  title={Movable Antennas Meet Low-Altitude Wireless Networks: Fundamentals, Opportunities, and Future Directions},
  author={Liu, Wenchao and Zhang, Xuhui and Wang, Chunjie and Ren, Jinke and Yuan, Weijie},
  journal={arXiv preprint arXiv:2506.13250},
  year={2025}
}

@ARTICLE{10654366,
  author={Liu, Wenchao and Zhang, Xuhui and Xing, Huijun and Ren, Jinke and Shen, Yanyan and Cui, Shuguang},
  journal={IEEE Wireless Commun. Lett.}, 
  title={{UAV}-Enabled Wireless Networks With Movable-Antenna Array: Flexible Beamforming and Trajectory Design}, 
  year={Mar. 2025},
  volume={14},
  number={3},
  pages={566-570},
}

@ARTICLE{10416363,
  author={Hu, Guojie and Wu, Qingqing and Xu, Kui and Si, Jiangbo and Al-Dhahir, Naofal},
  journal={IEEE Signal Process. Lett.}, 
  title={Secure Wireless Communication via Movable-Antenna Array}, 
  year={2024},
  volume={31},
  number={},
  pages={516-520},}

@ARTICLE{10663924,
  author={Pala, Sonia and Katwe, Mayur and Singh, Keshav and Tsiftsis, Theodoros A. and Li, Chih-Peng},
  journal={IEEE Trans. Veh. Technol.}, 
  title={Robust Transmission Design for {RIS}-Aided Full-Duplex-{RSMA} {V2X} Communications via Multi-Agent {DRL}}, 
  year={Jan. 2025},
  volume={74},
  number={1},
  pages={761-775},}

@ARTICLE{9676676,
  author={Xu, Yongjun and Xie, Hao and Wu, Qingqing and Huang, Chongwen and Yuen, Chau},
  journal={IEEE Trans. Commun.}, 
  title={Robust Max-Min Energy Efficiency for {RIS}-Aided {HetNets} With Distortion Noises}, 
  year={Feb. 2022},
  volume={70},
  number={2},
  pages={1457-1471},}

@ARTICLE{10666854,
  author={Lyu, Wanting and Yang, Songjie and Xiu, Yue and Chen, Xinyi and Zhang, Zhongpei and Assi, Chadi and Yuen, Chau},
  journal={IEEE Wireless Commun. Lett.}, 
  title={Dual-Robust Integrated Sensing and Communication: Beamforming Under {CSI} Imperfection and Location Uncertainty}, 
  year={Nov. 2024},
  volume={13},
  number={11},
  pages={3124-3128},}

@ARTICLE{9110587,
  author={Zhou, Gui and Pan, Cunhua and Ren, Hong and Wang, Kezhi and Renzo, Marco Di and Nallanathan, Arumugam},
  journal={IEEE Wireless Commun. Lett.}, 
  title={Robust Beamforming Design for Intelligent Reflecting Surface Aided {MISO} Communication Systems}, 
  year={Oct. 2020},
  volume={9},
  number={10},
  pages={1658-1662},}

@ARTICLE{7547360,
  author={Sun, Ying and Babu, Prabhu and Palomar, Daniel P.},
  journal={IEEE Trans. Signal Process.}, 
  title={Majorization-Minimization Algorithms in Signal Processing, Communications, and Machine Learning}, 
  year={Feb. 2017},
  volume={65},
  number={3},
  pages={794-816}}

@ARTICLE{11359731,
  author={Wang, Chunjie and Zhang, Xuhui and Liu, Wenchao and Ren, Jinke and Wang, Shuqiang and Shen, Yanyan and Ye, Kejiang and Tsang, Kim Fung},
  journal={IEEE Trans. Consum. Electron.}, 
  title={{RSMA}-Enhanced Data Collection in {RIS}-Assisted Intelligent Consumer Transportation Systems}, 
  year={to appear, 2026},
  volume={},
  number={},
  pages={},
}

@ARTICLE{11396947,
  author={Wang, Chunjie and Zhang, Xuhui and Liu, Wenchao and Ren, Jinke and Xing, Huijun and Wang, Shuqiang and Shen, Yanyan and Ye, Kejiang},
  journal={IEEE Trans. Consum. Electron.}, 
  title={Joint Beamforming and Resource Coordination for {RIS}-Aided {ISAC} over Secure {LAWNs}}, 
  year={to appear, 2026},
  volume={},
  number={},
  pages={},
}

@ARTICLE{11328802,
  author={Zhang, Xuhui and Liu, Wenchao and Wang, Chunjie and Ren, Jinke and Xing, Huijun and Wang, Shuqiang and Shen, Yanyan},
  journal={IEEE Trans. Consum. Electron.}, 
  title={{UAV}-Enabled Fluid Antenna Systems for Multi-Target Wireless Sensing over {LAWCNs}}, 
  year={to appear, 2026},
  volume={},
  number={},
  pages={},
}

@ARTICLE{9355403,
  author={Xu, Yongjun and Gui, Guan and Gacanin, Haris and Adachi, Fumiyuki},
  journal={IEEE Commun. Surveys Tuts.}, 
  title={A Survey on Resource Allocation for {5G} Heterogeneous Networks: Current Research, Future Trends, and Challenges}, 
  year={2nd Quart. 2021},
  volume={23},
  number={2},
  pages={668-695},
}

@ARTICLE{9293148,
  author={Hong, Sheng and Pan, Cunhua and Ren, Hong and Wang, Kezhi and Chai, Kok Keong and Nallanathan, Arumugam},
  journal={IEEE Trans. Wireless Commun.}, 
  title={Robust Transmission Design for Intelligent Reflecting Surface-Aided Secure Communication Systems With Imperfect Cascaded {CSI}}, 
  year={Apr. 2021},
  volume={20},
  number={4},
  pages={2487-2501},
}

@ARTICLE{11391523,
  author={Xu, He and Wu, Tuo and Tian, Ye and Jin, Ming and Liu, Wei and Guo, Qinghua and Elkashlan, Maged and Valenti, Matthew C. and Chae, Chan-Byoung and Tong, Kin-Fai and Wong, Kai-Kit},
  journal={IEEE Trans. Wireless Commun.}, 
  title={The Future Is Fluid: Revolutionizing {DOA} Estimation With Sparse Fluid Antennas}, 
  year={2026},
  volume={25},
  number={},
  pages={11546-11561},
}

@ARTICLE{9583590,
  author={Raza, Salman and Wang, Shangguang and Ahmed, Manzoor and Anwar, Muhammad Rizwan and Mirza, Muhammad Ayzed and Khan, Wali Ullah},
  journal={IEEE Internet Things J.}, 
  title={Task Offloading and Resource Allocation for {IoV} Using {5G} {NR}-{V2X} Communication}, 
  year={Jul. 2022},
  volume={9},
  number={13},
  pages={10397-10410},
}

@INPROCEEDINGS{10012694,
  author={Zhang, Xuhui and Xing, Huijun and Zang, Weilin and Jin, Zhenzhen and Shen, Yanyan},
  booktitle={Proc. IEEE Veh. Technol. Conf. (VTC)}, 
  title={Cybertwin-Driven Multi-Intelligent Reflecting Surfaces aided Vehicular Edge Computing Leveraged by Deep Reinforcement Learning}, 
  year={London, UK, Sep. 2022},
  volume={},
  number={},
  pages={1-7}
}

@ARTICLE{10354525,
  author={Zhao, Wei and Shi, Ke and Liu, Zhi and Wu, Xuangou and Zheng, Xiao and Wei, Linna and Kato, Nei},
  journal={IEEE Trans. Mobile Comput.}, 
  title={{DRL} Connects {Lyapunov} in Delay and Stability Optimization for Offloading Proactive Sensing Tasks of {RSUs}}, 
  year={Jul. 2024},
  volume={23},
  number={7},
  pages={7969-7982},
}

@INPROCEEDINGS{10946517,
  author={Zhang, Tiexin and Wu, Yinyu and Liu, Wenchao and Wang, Chunjie and Jin, Zhenzhen and Zhang, Xuhui and Shen, Yanyan},
  booktitle={Proc. IEEE Int. Conf. Commun. Technol. (ICCT)}, 
  title={Utility-Efficient Edge-Cloud Collaborative Space-Air-Ground Integrated Internet of Things}, 
  year={Chengdu, China, Oct. 2024},
  volume={},
  number={},
  pages={389-394},
}

@ARTICLE{10972043,
  author={Zhang, Xuhui and Liu, Wenchao and Ren, Jinke and Xing, Huijun and Gui, Gui and Shen, Yanyan and Cui, Shuguang},
  journal={IEEE Internet Things J.}, 
  title={Latency Minimization for {UAV}-Enabled Federated Learning: Trajectory Design and Resource Allocation}, 
  year={Jul. 2025},
  volume={12},
  number={14},
  pages={27097-27112},
}

@ARTICLE{10736570,
  author={Liu, Zhang and Du, Hongyang and Lin, Junzhe and Gao, Zhibin and Huang, Lianfen and Hosseinalipour, Seyyedali and Niyato, Dusit},
  journal={IEEE Trans. Mobile Comput.}, 
  title={{DNN} Partitioning, Task Offloading, and Resource Allocation in Dynamic Vehicular Networks: A {Lyapunov}-Guided Diffusion-Based Reinforcement Learning Approach}, 
  year={Mar. 2025},
  volume={24},
  number={3},
  pages={1945-1962},
}

@ARTICLE{10980172,
  author={Zhang, Xuhui and Xing, Huijun and Shen, Yanyan and Xu, Jie and Cui, Shuguang},
  journal={IEEE Trans. Wireless Commun.}, 
  title={Age of Information Minimization in {UAV}-Enabled {IoT} Networks via Federated Reinforcement Learning}, 
  year={Sep. 2025},
  volume={24},
  number={9},
  pages={7923-7939},
}

@ARTICLE{11298206,
  author={Chen, Yiyang and Liu, Wenchao and Zhang, Xuhui and Ren, Jinke and Xing, Huijun and Wang, Shuqiang and Shen, Yanyan and Tsang, Kim-Fung},
  journal={IEEE Trans. Consum. Electron.}, 
  title={Full-Duplex Integrated Sensing, Communication, and Computation over Low-Altitude Wireless Networks}, 
  year={to appear, 2026},
  volume={},
  number={},
  pages={},
}

@ARTICLE{9650760,
  author={Wong, Kai-Kit and Tong, Kin-Fai},
  journal={IEEE Trans. Wireless Commun.}, 
  title={Fluid Antenna Multiple Access}, 
  year={Jul. 2022},
  volume={21},
  number={7},
  pages={4801-4815},
}

@ARTICLE{10278220,
  author={Zhu, Lipeng and Ma, Wenyan and Zhang, Rui},
  journal={IEEE Commun. Lett.}, 
  title={Movable-Antenna Array Enhanced Beamforming: Achieving Full Array Gain With Null Steering}, 
  year={Dec. 2023},
  volume={27},
  number={12},
  pages={3340-3344},
}

@INPROCEEDINGS{10693833,
  author={Kuang, Ziming and Liu, Wenchao and Wang, Chunjie and Jin, Zhenzhen and Ren, Jinke and Zhang, Xuhui and Shen, Yanyan},
  booktitle={Proc. IEEE/CIC Int. Conf. Commun. China Workshops (ICCC Workshops)}, 
  title={Movable-Antenna Array Empowered {ISAC} Systems for Low-Altitude Economy}, 
  year={Hangzhou, China, Aug. 2024},
  volume={},
  number={},
  pages={776-781},
}

@ARTICLE{11108293,
  author={Dong, Xiangyu and Lyu, Wanting and Yang, Ran and Xiu, Yue and Mei, Weidong and Zhang, Zhongpei},
  journal={IEEE Commun. Lett.}, 
  title={Movable Antenna Enhanced Secure Simultaneous Wireless Information and Power Transfer}, 
  year={Oct. 2025},
  volume={29},
  number={10},
  pages={2356-2360},
}

@ARTICLE{11156108,
  author={Feng, Zhiyong and Zhao, Yujia and Yu, Kan and Li, Dong},
  journal={IEEE Trans. Commun.}, 
  title={Movable Antenna Empowered {PLS} With Eve’s Location Uncertainty: Joint Optimization of Beamforming and Antenna Positions}, 
  year={Dec. 2025},
  volume={73},
  number={12},
  pages={13708-13724},
}

@ARTICLE{11240557,
  author={Zhang, Xuhui and Liu, Wenchao and Ren, Jinke and Wang, Chunjie and Xing, Huijun and Shen, Yanyan and Cui, Shuguang},
  journal={IEEE Trans. Netw. Sci. Eng.}, 
  title={Movable-Antenna Empowered {AAV}-Enabled Data Collection over Low-Altitude Wireless Networks}, 
  year={2026},
  volume={13},
  number={},
  pages={4506-4523},
}

@ARTICLE{11261377,
  author={Yang, Hanyu and Xing, Chengwen and Gong, Shiqi and Ju, Xin and Zhao, Nan and Niyato, Dusit},
  journal={IEEE Trans. Wireless Commun.}, 
  title={A Framework for Energy-Efficiency Optimization in {MA}-Aided {MU}-{MIMO} Systems}, 
  year={2026},
  volume={25},
  number={},
  pages={7551-7568},
}

@ARTICLE{11389911,
  author={Wu, Tuo and Yao, Junteng and Zheng, Jianchao and Zhi, Kangda and Li, Xingwang and Elkashlan, Maged and Al-Dhahir, Naofal and Valenti, Matthew C. and Yuen, Chau},
  journal={IEEE Trans. Commun.}, 
  title={Unleashing More Potential from {FAS}: A Framework of {FAS}-{CoNOMA} Systems}, 
  year={to appear, 2026},
  volume={},
  number={},
  pages={},
}

@ARTICLE{9151971,
  author={Chen, Yingyang and Wen, Miaowen and Wang, Li and Liu, Weiping and Hanzo, Lajos},
  journal={IEEE Trans. Commun.}, 
  title={{SINR}-Outage Minimization of Robust Beamforming for the Non-Orthogonal Wireless Downlink}, 
  year={Nov. 2020},
  volume={68},
  number={11},
  pages={7247-7257},
}

@ARTICLE{9184012,
  author={Hu, Xiaoling and Zhong, Caijun and Alouini, Mohamed-Slim and Zhang, Zhaoyang},
  journal={IEEE Wireless Commun. Lett.}, 
  title={Robust Design for {IRS}-Aided Communication Systems With User Location Uncertainty}, 
  year={Jan. 2021},
  volume={10},
  number={1},
  pages={63-67},
}

@ARTICLE{9845399,
  author={Chen, Yun and Zhang, Guoping and Xu, Hongbo and Ren, Yinshuan and Chen, Xue and Li, Ruijie},
  journal={IEEE Commun. Lett.}, 
  title={Robust Beamforming and Power Allocation for Secure Communication in Systems With Imperfect Channel and Hardware Impairments}, 
  year={Oct. 2022},
  volume={26},
  number={10},
  pages={2277-2281},
}

@ARTICLE{10027173,
  author={Liao, Bin and Xiong, Xue and Quan, Zhi},
  journal={IEEE Trans. Veh. Technol.}, 
  title={Robust Beamforming Design for Dual-Function Radar-Communication System}, 
  year={Jun. 2023},
  volume={72},
  number={6},
  pages={7508-7516},
}

@ARTICLE{10513353,
  author={Dai, Jianxin and Ye, Jianglin and Wang, Kezhi and Pan, Cunhua and Fan, Hui},
  journal={IEEE Wireless Commun. Lett.}, 
  title={Joint Radar-Communication Beamforming Considering Both Transceiver Hardware Impairments and Imperfect {CSI}}, 
  year={Jul. 2024},
  volume={13},
  number={7},
  pages={1898-1902},
}

@ARTICLE{11192669,
  author={Ni, Hao and Xu, Yongjun and Li, Xingwang and Yuan, Yiming and Yuen, Chau and Chen, Li},
  journal={IEEE Trans. Veh. Technol.}, 
  title={Robust Resource Allocation for Movable Antenna-Enabled {MISO} Communication Systems with Angular Uncertainties}, 
  year={to appear, 2026},
  volume={},
  number={},
  pages={},
}

@ARTICLE{11107317,
  author={Lyu, Bin and Liu, Hao and Hua, Meng and Hong, Wenqing and Gong, Shimin and Tian, Feng and Jamalipour, Abbas},
  journal={IEEE Trans. Wireless Commun.}, 
  title={Robust Transmission Design for Reconfigurable Intelligent Surface and Movable Antenna Enabled Symbiotic Radio Communications}, 
  year={2026},
  volume={25},
  number={},
  pages={1702-1716},
}

@ARTICLE{10012421,
  author={Wei, Zhiqing and Qu, Hanyang and Wang, Yuan and Yuan, Xin and Wu, Huici and Du, Ying and Han, Kaifeng and Zhang, Ning and Feng, Zhiyong},
  journal={IEEE Internet Things J.}, 
  title={Integrated Sensing and Communication Signals Toward {5G-A} and {6G}: A Survey}, 
  year={Jul. 2023},
  volume={10},
  number={13},
  pages={11068-11092},
}

@INPROCEEDINGS{11374087,
  author={Wang, Yike and Wu, Zhike and Chen, Jiang and Wang, Chunjie and Zhang, Xuhui and Shen, Yanyan},
  booktitle={Proc. IEEE Int. Conf. Commun. Technol. (ICCT)}, 
  title={QoS-Aware Integrated Sensing, Communication, and Control with Movable Antenna}, 
  year={Shenyang, China, Oct. 2025},
  volume={},
  number={},
  pages={565-570},
}

@article{liu2025uav,
  title={{UAV}-Enabled {ISAC} Systems with Fluid Antennas},
  author={Liu, Wenchao and Zhang, Xuhui and Ren, Jinke and Yuan, Weijie and You, Changsheng and Li, Shuangyang},
  journal={arXiv preprint arXiv:2509.21105},
  year={2025}
}

@INPROCEEDINGS{11374005,
  author={Chen, Yiyang and Liu, Wenchao and Wang, Chunjie and Wu, Yinyu and Zhang, Xuhui and Shen, Yanyan},
  booktitle={Proc. IEEE Int. Conf. Commun. Technol. (ICCT)}, 
  title={Latency Minimization for Multi-{AAV}-Enabled {ISCC} Systems with Movable Antenna}, 
  year={Shenyang, China, Oct. 2025},
  volume={},
  number={},
  pages={1232-1238},
}

@article{wang2025joint,
  title={Joint Beamforming Design for {RIS}-Empowered {NOMA}-{ISAC} Systems},
  author={Wang, Chunjie and Zhang, Xuhui and Ren, Jinke and Liu, Wenchao and Wang, Shuqiang and Shen, Yanyan and Ye, Kejiang and Xu, Chengzhong and Niyato, Dusit},
  journal={arXiv preprint arXiv:2508.13842},
  year={2025}
}
\end{document}